%% file: main.tex
\documentclass[journal]{IEEEtran}
\ifCLASSINFOpdf
\else
\fi

\usepackage{cite}
\usepackage{amsmath,amssymb,amsfonts}
\usepackage{algorithmic}
\usepackage{graphicx}
\usepackage{textcomp}
\usepackage{xcolor}
\usepackage{amsthm}
\usepackage{balance} 
\usepackage{tikz} % Useful for drawing images, used for creating the frontpage
\usetikzlibrary{positioning} % Additional library for relative positioning 
\usetikzlibrary{calc,shapes} % Additional library for calculating within tikz
\usepackage{tikz-3dplot}
\usepackage{clipboard}
\usepackage{xr}
\newtheorem{lemma}{Lemma}
\newtheorem{proposition}{Proposition}
\newtheorem{remark}{Remark}
\newtheorem{corollary}{Corollary}

\begin{document}
\title{Beamfocusing Capabilities of a Uniform Linear Array in the Holographic Regime}

\author{Xavier~Mestre,~\IEEEmembership{Senior Member,~IEEE,}
        Adrian~Agustin,~\IEEEmembership{Senior Member,~IEEE}%~\IEEEmembership{Life~Fellow,~IEEE}% <-this % stops a space
\thanks{This work is supported by the grant from the Spanish ministry of economic affairs and digital transformation and of the European union - NextGenerationEU UNICO-5G I+D/AROMA3D-Earth and AROMA3D-Hybrid (TSI-063000-2021-69 and TSI-063000-2021-71), by Grant 2021 SGR 00772 funded by the Universities and Research Department from Generalitat de Catalunya, and by the Spanish Government through the project 6G AI-native Air Interface (6G-AINA, PID2021-128373OB-I00 funded by MCIN/AEI/10.13039/501100011033) ERDF A way of making Europe. }% <-this % stops a space
}



\markboth{Submitted to IEEE Transactions on Signal Processing}{Submitted to IEEE Transactions on Signal Processing}



\maketitle

\begin{abstract}
The use of multiantenna technologies in the near field offers the possibility of focusing the energy in spatial regions rather than just in angle. The objective of this paper is to provide a formal framework that allows to establish the region in space where this effect can take place and how efficient this focusing can be, assuming that the transmit architecture is a uniform linear array (ULA). A dyadic Green's channel model is adopted, and the amplitude differences between the receiver and each transmit antenna are effectively incorporated in the model. By considering a second-order expansion of the signal to noise ratio (SNR) around the intended receiver, a formal criterion is derived in order to establish whether beamfocusing is feasible or not. For the regions where beamfocusing is indeed possible, an analytic description is provided that determines the shape and position of the asymptotic ellipsoid where a minimum SNR is achieved. Further insights are provided by considering the holographic regime, whereby the number of elements of the ULA increase without bound while the distance between adjacent elements converges to zero. This asymptotic framework allows to simplify the analytical form of the beamfocusing feasibility region, which in turn provides some further insights into the shape of the coverage regions depending on the position of the intended receiver. In particular, it is shown that beamfocusing is only possible if the size of the ULA is at least $4.4\lambda$ where $\lambda$ is the transmission wavelength. Furthermore, a closed form analytical expression is provided that asymptotically determines the maximum distance where beamfocusing is feasible as a function of the elevation angle. In particular, beamfocusing is only feasible when the receiver is located between a minimum and a maximum distance from the array, where these upper and lower distance limits effectively depend on the angle of elevation. 
\end{abstract}


% \fi
\IEEEpeerreviewmaketitle

\section{Introduction}

Modern wireless communication systems heavily depend on multi-antenna technology to maximize performance in terms of reliability and throughput. As these systems shift to higher frequency bands, such as millimeter wave (mmWave) and sub-terahertz (THz) frequencies \cite{Priebe14,Shafi18}, multi-antenna technology becomes crucial for overcoming the more challenging propagation conditions. These higher frequency bands are characterized by significant propagation losses and increased material absorption, leading to communication channels with a dominant line-of-sight component. 
This introduces challenges for traditional multi-antenna signal processing techniques, which must be redesigned in order to effectively utilize the inherent structure of these channels. Simultaneously, the adoption of extra-large antenna arrays is growing to satisfy the high-performance requirements of future wireless networks These extensive antenna setups are approaching the size of the transmitter-receiver distance, which is itself being reduced to mitigate high propagation losses. This marks a fundamental shift from conventional wireless architectures, which typically assume operation in the far field. In these evolving scenarios, the assumption of planar electromagnetic wave fronts is no longer valid, necessitating the adoption of more sophisticated near-field channel models. Interestingly, this near-field structure offers potential performance advantages in line-of-sight conditions \cite{Driessen99,Jiang05,Bohagen09,Sanguinetti23}. 

One of the main advantages of operating in the near field is the fact that spatial filters may be used to do \emph{beamfocusing}, as opposed to conventional beamforming (in the far field). This typically refers to the capability of focusing the energy of the transmitted signal towards a specific spatial location rather than just an angular direction (as beamforming). In order words, a spatial filter may be designed so that it targets a specific point in space \cite{Zhang23,bjornson2024towards,Liu23,Ramezani2024} rather than just a direction of radiation. This opens up the possibility of multiplexing communications no only in angle but also in location, so that multiple user equipments could in principle share the same radio link even if they are all seen from the same direction of arrival \cite{Myers22,Wu23}. 

A lot of research has recently taken place aimed at the characterization of the beamfocusing phenomenon. Most of the literature has been focused on the characterization of the spatial region where beamfocusing is feasible (i.e. the beamfocusing feasibility region) as well as the determination of the spatial resolution of the spatial response (both in angle and depth). Research so far has considered simplified channel models that essentially disregard the amplitude variations across the array \cite{kosasih24,Bjornson21asilomar}, an approximation that is less valid when the dimensions of the array are comparable in magnitude with the propagation distance. This assumption is dropped in the present contribution, where a slightly more general channel model based on the Green function is adopted. This is inspired by work in \cite{Wei23}, which considers a multi-user MIMO setting based on a channel model with three spatial polarizations based on the Green dyadic function. This channel model has also been recently proposed in \cite{gong2024near} to analyze the capacity of holographic surfaces composed of densely-packed surface radiators with arbitrary placements, as well as in \cite{agustin23}, where infinitesimal dipoles with three orthogonal polarizations were considered on both ends of the communication link.

The Green dyadic channel model \cite{Poon05} represents the situation in which the radiating elements are composed of up to three orthogonal infinitesimal dipoles. This can be taken as an approximation of a more realistic situation in which the radiating elements have a more complicated radiation pattern. It can also be used to model the behavior of arrays of metasurface antennas \cite{Imani20,Lipworth13,Lipworth15,Smith17,Imani18}. These are the basic implementation of holographic apertures \cite{Dardari21,Huang20,Gong24survey,gong2024near,Pizzo22}, also known as continuous aperture architectures \cite{Sayeed10,Sayeed11} or large intelligent surfaces \cite{Hu18} in the array processing literature. Holographic apertures act as a continuous antenna array, dynamically shaping and controlling signal generation and electromagnetic wave projection, effectively creating "holograms" of electromagnetic fields \cite{Gong24survey}.
The behavior of such holographic apertures can be accurately modeled by considering an asymptotically large number of antennas with an asymptotically small inter-element spacing. This approach treats the array aperture as a continuum of radiating elements capable of manipulating electromagnetic waves at the most fundamental physical level. 

The present paper builds upon the approach in \cite{agustin24twc}, which proposed the use of the asymptotic holographic description of a general uniform planar array to characterize the achievable rate of a single-user channel employing multiple polarizations. In this paper, we will employ the same channel model and approach in order to further investigate the beamfocusing capabilities of a ULA in the holographic regime. The study allows to conclude that beamfocusing is only possible when the total length of the ULA is at least $4.4\lambda$ where $\lambda$ is the transmission wavelength. Simple analytical formulas are also provided that allow to describe the range of distances where beamfocusing is possible for each given elevation angle. 

\section{Scenario and Signal Model}

Consider a downlink transmission from a uniform linear array (ULA) towards a certain intended receiver\footnote{Note that most of the conclusions also hold in the reverse (uplink) situation due to the symmetry of the problem.}.
The transmit ULA is composed of $2M+1$ elements placed along the
$y$-axis with an interelement separation of $\Delta_{T}$ meters, see further Fig.~\ref{fig:Scenario}. Each element
of the ULA consists of up to three infinitesimal dipoles aligned with the
three cartesian coordinates. We will denote as $t_{\mathrm{pol}}$ the number
of polarizations employed at the transmit array, so that when $t_{\mathrm{pol}
}=1$ (resp. $t_{\mathrm{pol}}=2$, $t_{\mathrm{pol}}=3$) each ULA\ element
employs one (resp. two, three) linear polarizations along the $x$-axis (resp.
$x,y$-axes, $x,y,z$-axes).

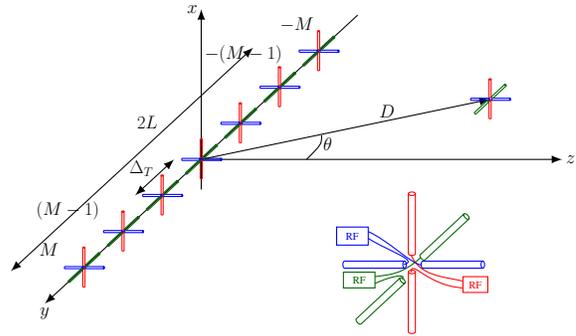
\begin{figure}[h] 
    \centering 
\centerline{\input{figures/ULA.tex}}
\caption{Scenario configuration. The transmitter is a ULA consisting of $2M+1$ elements, each incorporating 3 orthogonal infinitesimal dipoles. It is assumed that at least the polarization in the $x$-axis (red dipoles) is employed at the transmitter. Here, $D$ is the distance between the receiver and the center of the array, $\theta$ is the elevation angle.}
    \label{fig:Scenario} 
\end{figure}

Let us now consider the case where only one symbol stream is transmitted, and
the whole transmit array is used to focalize the energy towards a certain
point with cartesian coordinates $\mathbf{r}_{0}\in\mathbb{R}^{3}$. The signal
model at a certain point with coordinates $\mathbf{r}\in\mathbb{R}^{3}$ can be
expressed as
\[
\mathbf{y}=\mathbf{H}_{\mathrm{pol}}\left(  \mathbf{r}\right)  \mathbf{w}
x+\mathbf{n}
\]
where $\mathbf{w}$ is the unit norm tramsmit spatial filter, of dimensions
$\left(  2M+1\right)  t_{\mathrm{pol}}$, $x$ is the transmit signal and
$\mathbf{n}$ is the noise vector (with zero mean and power $\sigma^{2}$). We
will assume that the receiver consists of three infinitesimal dipoles that are
independently downconverted/sampled, so that both $\mathbf{y}$ and
$\mathbf{n}$ have dimensions $r_{\mathrm{pol}}=3$ in this study. The channel
matrix $\mathbf{H}_{\mathrm{pol}}\left(  \mathbf{r}\right)  $ can be
decomposed as
\begin{equation}
\mathbf{H}_{\mathrm{pol}}\left(  \mathbf{r}\right)  =\left[  \mathbf{H}
_{-M}^{t_{\mathrm{pol}}\times r_{\mathrm{pol}}}\left(  \mathbf{r}\right)
,\ldots,\mathbf{H}_{M}^{t_{\mathrm{pol}}\times r_{\mathrm{pol}}}\left(
\mathbf{r}\right)  \right]  \label{eq:HpolFromHm}
\end{equation}
where $\mathbf{H}_{m}^{t_{\mathrm{pol}}\times r_{\mathrm{pol}}}\left(
\mathbf{r}\right)  =\left[  \mathbf{H}_{m}\left(  \mathbf{r}\right)  \right]
_{1:r_{\mathrm{pol}}\text{,}1:t_{\mathrm{pol}}}$ and where $\mathbf{H}
_{m}\left(  \mathbf{r}\right)  $ is a $3\times3$ channel matrix between the
three polarizations at the $m$th element of the\ ULA and the three
polarizations at the receiver. Disregarding the reactive terms and considering
only the radiative electromagnetic field, the channel model can be expressed
as$\ \mathbf{H}_{m}\left(  \mathbf{r}\right)  =h_{m}\left(  \mathbf{r}\right)
\mathbf{P}_{m}^{\mathbf{\perp}}(\mathbf{r})$ with
\begin{equation}
h_{m}\left(  \mathbf{r}\right)  =\frac{\xi}{\lambda}\frac{1}{\left\Vert
\mathbf{r-p}_{m}\right\Vert }\exp\left(  -\mathrm{j}\frac{2\pi}{\lambda
}\left\Vert \mathbf{r-p}_{m}\right\Vert \right)  \label{eq:defLowCaseh}
\end{equation}
where $\xi$ is a certain complex constant\footnote{If the channel is
considered to be proportional to the electromagnetic field, the constant $\xi$
can be seen to be proportional to the permittivity of the medium.}, $\lambda
$\ is the transmission wavelength, $\mathbf{p}_m$ is position vector of the $m$th element of the ULA and where $\mathbf{P}_{m}^{\mathbf{\perp}
}(\mathbf{r})$ is the orthogonal projection matrix
\[
\mathbf{P}_{m}^{\mathbf{\perp}}(\mathbf{r})=\mathbf{I}_{3}-\mathbf{P}
_{m}(\mathbf{r}),\text{\quad}\mathbf{P}_{m}(\mathbf{r})=\frac{\left(
\mathbf{r-p}_{m}\right)  \left(  \mathbf{r-p}_{m}\right)  ^{T}}{\left\Vert
\mathbf{r-p}_{m}\right\Vert ^{2}}.
\]
Consider the singular value decomposition
\[
\frac{1}{\sqrt{2M+1}}\mathbf{H}_{\mathrm{pol}}\left(  \mathbf{r}_{0}\right)
=\mathbf{U\Sigma}^{1/2}\mathbf{V}^{H}.
\]
Let $\mathbf{u}$ and $\mathbf{v}$ respectively denote the left and right
singular vector associated to the maximum singular value. It is well known
that the transmit filter $\mathbf{w=v}$ is the one that maximizes the SNR\ at
a receiver located at $\mathbf{r}=\mathbf{r}_{0}$. We will assume from now on
that this is the spatial filter that is used at the transmit side.

Assume that the transmit power is fixed to $\mathbb{E}\left\vert x\right\vert
^{2}=\bar{P}/(2M+1)$, i.e. it decreases linearly with the number of transmit
antennas. We define the received signal to noise ratio at the position
$\mathbf{r}$ as
\begin{equation}
\mathsf{SNR}(\mathbf{r})=\bar{P}\frac{\mathbf{w}^{H}\mathbf{H}_{\mathrm{pol}
}^{H}\left(  \mathbf{r}\right)  \mathbf{H}_{\mathrm{pol}}\left(
\mathbf{r}\right)  \mathbf{w}}{\sigma^{2}\left(  2M+1\right)  }.
\label{eq:SNRdef}
\end{equation}
In particular, $\mathsf{SNR}(\mathbf{r}_{0})\,$\ is the SNR at the intended
location, that is
\[
\mathsf{SNR}(\mathbf{r}_{0})=\frac{\bar{P}\lambda_{\max}}{\sigma^{2}}
\]
where $\lambda_{\max}$ is the maximum eigenvalue of $\mathbf{H}_{\mathrm{pol}
}\left(  \mathbf{r}_{0}\right)  \mathbf{H}_{\mathrm{pol}}^{H}\left(
\mathbf{r}_{0}\right)  /(2M+1)$. The objective of this paper is to study how $\mathsf{SNR}(\mathbf{r})$ behaves near the intended position $\mathbf{r}_0$. To that effect, we will carry out a perturbation analysis of $\mathsf{SNR}(\mathbf{r})$ around this point. We will say that beamfocusing is feasible at $\mathbf{r}_0$ when $\mathsf{SNR}(\mathbf{r})$ decays locally as we move away from the intended position. Hence, the beamfocusing feasibility region is defined as the region of points in space such that $\mathsf{SNR}(\mathbf{r})$ behave as a locally concave function. This approach is different from previous studies, which focused on the characterization of the \emph{global} spatial response. Our approach is, in this sense, eminently local. 

\begin{remark} \label{rem:niceConstantsPols}
Here and throughout this paper we will need to bound a number of error terms
by quantities that are not altered when considering the holographic regime
(that is when the distance between infinitesimal dipoles goes to zero). To that
effect, we define a \emph{``nice constant"} as a constant that depends only on
the quantities $\left\vert \xi\right\vert $, $\bar{P}$, $\sigma^{2}$,
$\lambda$. In particular, a nice constant does not depend on the transmit
antenna index ($m$), the number of elements of the ULA\ ($M$), the position of
the different array elements at the ULA ($\mathbf{p}_{m}$) or the position of
the receiver ($\mathbf{r}$, $\mathbf{r}_{0}$). Likewise, we define a
\emph{nice polynomial} as a polynomial whose coefficients are nice constants.
Throughout the paper, nice polynomials will generally be denoted as
$\mathrm{P}($\textperiodcentered$)$. Note, however, that the actual value of
the polynomial may vary from one line to another.
\end{remark}

The following proposition establishes a second order Taylor series
approximation of the channel matrix. We will use the definition of nice
polynomial above to provide a bound on the reminder term that will be employed
later on. The objective of this bound is to ensure that the error associated  
to the reminder term in the expansion does not scale up when considering the 
holographic approximation. 

\begin{proposition}
\label{prop:Taylor}The channel matrix $\mathbf{H}_{m}\left(  \mathbf{r}
\right)  $ accepts the following Taylor expansion around $\mathbf{r}
=\mathbf{r}_{0}$
\[
\mathbf{H}_{m}\left(  \mathbf{r}\right)  =\mathbf{H}_{m}\left(  \mathbf{r}
_{0}\right)  +\mathcal{G}_{m}\left(  \mathbf{\Delta}_{r}\right)
+\mathcal{H}_{m}\left(  \mathbf{\Delta}_{r}\right)  +\mathcal{R}
_{m}(\mathbf{\bar{r},\Delta}_{r})
\]
where $\mathbf{\Delta}_{r}=\left(  \mathbf{r}-\mathbf{r}_{0}\right)  =\left[
\Delta_{x},\Delta_{y},\Delta_{z}\right]  ^{T}$ and where $\mathbf{\bar{r}}$ is
in the segment joining $\mathbf{r}$ and $\mathbf{r}_{0}$. The three matrices
$\mathcal{G}_{m}\left(  \mathbf{\Delta}_{r}\right)  $, $\mathcal{H}_{m}\left(
\mathbf{\Delta}_{r}\right)  $ and $\mathcal{R}_{m}(\mathbf{\bar{r},\Delta}
_{r})$ respectively correspond to the gradient, Hessian and reminder terms,
and can be described as follows. The gradient term $\mathcal{G}_{m}\left(
\mathbf{\Delta}_{r}\right)  $ takes the form
\begin{align}
\mathcal{G}_{m}\left(  \mathbf{\Delta}_{r}\right)   &  =-\left(
1+\mathrm{j}\frac{2\pi}{\lambda}\left\Vert \mathbf{r}_{m,0}\right\Vert
\right)  \frac{\Delta_{r}^{T}\mathbf{r}_{m,0}}{\left\Vert \mathbf{r}
_{m,0}\right\Vert ^{2}}\mathbf{H}_{m}(\mathbf{r}_{0}) \label{eq:defGm(DeltaT)}
\\
&  -\frac{\mathbf{r}_{m,0}\Delta_{r}^{T}}{\left\Vert \mathbf{r}_{m,0}
\right\Vert ^{2}}\mathbf{H}_{m}(\mathbf{r}_{0})-\mathbf{H}_{m}(\mathbf{r}
_{0})\frac{\Delta_{r}\mathbf{r}_{m,0}^{T}}{\left\Vert \mathbf{r}
_{m,0}\right\Vert ^{2}}\nonumber
\end{align}
where $\mathbf{r}_{m,0}=\mathbf{r}_{0}-\mathbf{p}_{m}$. The Hessian term
$\mathcal{H}_{m}\left(  \mathbf{\Delta}_{r}\right)  $ can be expressed as
\begin{align}
\mathcal{H}_{m}\left(  \mathbf{\Delta}_{r}\right)  & =\frac{h_{m}(\mathbf{r}
_{0})}{\left\Vert \mathbf{r}_{m,0}\right\Vert ^{2}}\Bigg[  \mathbf{\Xi}
_{1,0}+\left(  1+\mathrm{j}\frac{2\pi}{\lambda}\left\Vert \mathbf{r}
_{m,0}\right\Vert \right)  \mathbf{\Xi}_{2,0} \nonumber  \\
& +\left(  1+\mathrm{j}\frac{2\pi
}{\lambda}\left\Vert \mathbf{r}_{m,0}\right\Vert \right)  ^{2}\mathbf{\Xi
}_{3,0}\Bigg] \label{eq:defHm(DeltaT)}
\end{align}
where the three matrices $\mathbf{\Xi}_{1,0}$, $\mathbf{\Xi}_{2,0}$,
$\mathbf{\Xi}_{3,0}$ take the form
\begin{align*}
\mathbf{\Xi}_{1,0}  &  =-2\mathbf{P}_{m}^{\perp}(\mathbf{r}_{0})\mathbf{\Delta
}_{r}\mathbf{\Delta}_{r}^{T}\mathbf{P}_{m}^{\perp}(\mathbf{r}_{0})\\
&  +2\mathbf{P}_{m}(\mathbf{r}_{0})\mathbf{\Delta}_{r}\mathbf{\Delta}_{r}
^{T}\mathbf{P}_{m}^{\perp}(\mathbf{r}_{0})+2\mathbf{P}_{m}^{\perp}
(\mathbf{r}_{0})\mathbf{\Delta}_{r}\mathbf{\Delta}_{r}^{T}\mathbf{P}
_{m}(\mathbf{r}_{0})\\
&  +\left(  \mathbf{\Delta}_{r}^{T}\mathbf{P}_{m}(\mathbf{r}_{0}
)\mathbf{\Delta}_{r}\right)  \mathbf{P}_{m}^{\perp}(\mathbf{r}_{0})+2\left(
\mathbf{\Delta}_{r}^{T}\mathbf{P}_{m}^{\perp}(\mathbf{r}_{0})\mathbf{\Delta
}_{r}\right)  \mathbf{P}_{m}(\mathbf{r}_{0})\\
\mathbf{\Xi}_{2,0}  &  =2\mathbf{P}_{m}(\mathbf{r}_{0})\mathbf{\Delta}
_{r}\mathbf{\Delta}_{r}^{T}\mathbf{P}_{m}^{\perp}(\mathbf{r}_{0}
)+2\mathbf{P}_{m}^{\perp}(\mathbf{r}_{0})\mathbf{\Delta}_{r}\mathbf{\Delta
}_{r}^{T}\mathbf{P}_{m}(\mathbf{r}_{0})\\
&  -\left(  \mathbf{\Delta}_{r}^{T}\mathbf{P}_{m}^{\perp}(\mathbf{r}
_{0})\mathbf{\Delta}_{r}\right)  \mathbf{P}_{m}^{\perp}(\mathbf{r}_{0})\\
\mathbf{\Xi}_{3,0}  &  =\left(  \mathbf{\Delta}_{r}^{T}\mathbf{P}
_{m}(\mathbf{r}_{0})\mathbf{\Delta}_{r}\right)  \mathbf{P}_{m}^{\perp
}(\mathbf{r}_{0}).
\end{align*}
Finally, the spectral norm of the reminder term $\mathcal{R}_{m}
(\mathbf{\bar{r},\Delta}_{r})$ can be upper bounded as
\[
\left\Vert \mathcal{R}_{m}(\mathbf{\bar{r},\Delta}_{r})\right\Vert
\leq\mathrm{P}\left(  \left\Vert \mathbf{\bar{r}-p}_{m}\right\Vert
^{-1}\right)  \left\Vert \mathbf{\Delta}_{r}\right\Vert ^{3}
\]
where $\mathrm{P}($\textperiodcentered$)$ is a nice polynomial (cf. definition in Remark~\ref{rem:niceConstantsPols}).
\end{proposition}

\begin{proof}
See Appendix \ref{sec:AppendixTaylor}.
\end{proof}

Based on the above result and mimicking the construction of the channel matrix
$\mathbf{H}_{\mathrm{pol}}\left(  \mathbf{r}\right)  $ in (\ref{eq:HpolFromHm}) from the individual $\mathbf{H}_{m}\left(  \mathbf{r}\right)  $, we can more
compactly write
\begin{equation}
\mathbf{H}_{\mathrm{pol}}\left(  \mathbf{r}\right)  =\mathbf{H}_{\mathrm{pol}
}\left(  \mathbf{r}_{0}\right)  +\mathcal{G}_{\mathrm{pol}}\left(
\mathbf{\Delta}_{r}\right)  +\mathcal{H}_{\mathrm{pol}}\left(  \mathbf{\Delta
}_{r}\right)  +\mathcal{R}_{\mathrm{pol}}(\mathbf{\bar{r},\Delta}_{r})
\label{eq:TaylorChannel}
\end{equation}
where $\mathcal{G}_{\mathrm{pol}}\left( \mathbf{ \Delta}_{r}\right)  $, $\mathcal{H}
_{\mathrm{pol}}\left(  \mathbf{\Delta}_{r}\right)  $ and $\mathcal{R}_{\mathrm{pol}
}(\mathbf{\bar{r},\mathbf{\Delta}}_{r})$ have dimensions $r_{\mathrm{pol}}\times\left(
2M+1\right)  t_{\mathrm{pol}}$ and are built from the corresponding quantities
in Proposition \ref{prop:Taylor} in the same way as (\ref{eq:HpolFromHm}).
Using (\ref{eq:TaylorChannel}) in (\ref{eq:SNRdef}) we can obtain the
following interesting result.

\begin{corollary}
\label{cor:TaylorSNR}The SNR in (\ref{eq:SNRdef}) accepts the expansion
\begin{equation}
\mathsf{SNR}(\mathbf{r})   =\mathsf{SNR}(\mathbf{r}_{0})  +\frac{\bar{P}
}{\sigma^{2}}\frac{\mathbf{w}^{H}\left( \ast \right)
\mathbf{w}}{2M+1} +\epsilon_{M}\label{eq:TaylorSNR} 
\end{equation}
where 
\begin{align} \label{eq:ast}
(\ast) & =  \mathcal{G}_{\mathrm{pol}}^{H}\left(
\mathbf{\Delta}_{r}\right)  \mathbf{H}_{\mathrm{pol}}\left(  \mathbf{r}
_{0}\right)  +\mathbf{H}_{\mathrm{pol}}^{H}\left(  \mathbf{r}_{0}\right)
\mathcal{G}_{\mathrm{pol}}\left(  \mathbf{\Delta}_{r}\right) 
\\ &+ \mathcal{H}
_{\mathrm{pol}}^{H}\left(  \mathbf{\Delta}_{r}\right)  \mathbf{H}
_{\mathrm{pol}}\left(  \mathbf{r}_{0}\right)    +\mathbf{H}_{\mathrm{pol}}^{H}\left(
\mathbf{r}_{0}\right)  \mathcal{H}_{\mathrm{pol}}\left(  \mathbf{\Delta}
_{r}\right) 
\\ & + \mathcal{G}_{\mathrm{pol}}
^{H}\left(  \mathbf{\Delta}_{r}\right)  \mathcal{G}_{\mathrm{pol}}\left(
\mathbf{\Delta}_{r}\right)
\end{align}
and where $\epsilon_{M}$ is a third-order error term and is such that $\left\vert
\epsilon_{M}\right\vert \leq K\left\Vert \mathbf{\Delta}_{r}\right\Vert ^{3}$
for a certain nice constant $K$.
\end{corollary}

\begin{proof}
See Appendix \ref{sec:AppendixTaylorSNR}.
\end{proof}

We recall that $\mathbf{w}$\ is chosen to be the left singular vector
associated to the maximum singular value of the channel matrix $\mathbf{H}
_{\mathrm{pol}}\left(  \mathbf{r}_{0}\right)  $. This implies that
$\mathbf{w}$ alternatively expressed as\footnote{To see this, simply use the
SVD formula $\mathbf{H}_{\mathrm{pol}}^{H}\left(  \mathbf{r}_{0}\right)
=\sqrt{2M+1}\mathbf{\mathbf{V}\Sigma}^{1/2}\mathbf{U}^{H}$.}
\begin{equation}
\mathbf{w}=\frac{1}{\sqrt{2M+1}\sqrt{\lambda_{\max}}}\mathbf{H}_{\mathrm{pol}
}^{H}\left(  \mathbf{r}_{0}\right)  \mathbf{u} \label{eq:wSimplified}
\end{equation}
where we recall that $\mathbf{u}$ is the eigenvector associated to the maximum
eigenvalue of $\mathbf{H}_{\mathrm{pol}}\left(  \mathbf{r}_{0}\right)
\mathbf{H}_{\mathrm{pol}}^{H}\left(  \mathbf{r}_{0}\right)  /(2M+1)$, denoted
as $\lambda_{\max}$. By inserting (\ref{eq:wSimplified}) into the expression
of $\mathsf{SNR}(\mathbf{r})$ in (\ref{eq:TaylorSNR}) we see that the
expression can be reformulated as
\begin{equation}
\mathsf{SNR}(\mathbf{r})   =\mathsf{SNR}(\mathbf{r}_{0})  +\frac{\bar{P}
}{\sigma^{2}}\mathbf{u}^{H}\left( \star \right)
\mathbf{u} +\epsilon_{M} \label{eq:TaylorSNRsimple}
\end{equation}
where now 
\begin{align*}
    \left( \star \right) & = \frac{\mathbf{H}_{\mathrm{pol}}\left(
\mathbf{r}_{0}\right)  \mathcal{G}_{\mathrm{pol}}^{H}\left(  \Delta
_{r}\right) + \mathcal{G}_{\mathrm{pol}}\left(
\Delta_{r}\right)  \mathbf{H}_{\mathrm{pol}}^{H}\left(  \mathbf{r}_{0}\right)}{2M+1}\\
&  + \frac{\mathbf{H}_{\mathrm{pol}
}\left(  \mathbf{r}_{0}\right)  \mathcal{H}_{\mathrm{pol}}^{H}\left(
\Delta_{r}\right)+\mathcal{H}_{\mathrm{pol}}\left(
\Delta_{r}\right)  \mathbf{H}_{\mathrm{pol}}^{H}\left(  \mathbf{r}_{0}\right)
}{2M+1}\nonumber\\
&  + \frac{\mathbf{H}_{\mathrm{pol}
}\left(  \mathbf{r}_{0}\right)  \mathcal{G}_{\mathrm{pol}}^{H}\left(
\Delta_{r}\right)  \mathcal{G}_{\mathrm{pol}}\left(  \Delta_{r}\right)
\mathbf{H}_{\mathrm{pol}}^{H}\left(  \mathbf{r}_{0}\right) 
}{\lambda_{\max}\left(  2M+1\right)  ^{2}}.
\end{align*}
Hence, we can investigate the behavior of $\mathsf{SNR}(\mathbf{r})$ around
$\mathsf{SNR}(\mathbf{r}_{0})$ by simply studying the two matrices
$\mathcal{G}_{\mathrm{pol}}\left(  \Delta_{r}\right)  \mathbf{H}
_{\mathrm{pol}}^{H}\left(  \mathbf{r}_{0}\right)  /\left(  2M+1\right)  $ and
$\mathcal{H}_{\mathrm{pol}}\left(  \Delta_{r}\right)  \mathbf{H}
_{\mathrm{pol}}^{H}\left(  \mathbf{r}_{0}\right)  /\left(  2M+1\right)  $.
These two matrices have a complicated analytical form, but we can
significantly simplify the result by assuming that the intended receiver is
located on the $yz$-plane, i.e. $x_{0}=0$. In order to present the results, we
introduce the following key quantities
\begin{align*}
s_{M}^{(k)}  &  =\frac{1}{2M+1}\sum_{m=-M}^{M}\frac{1}{\left\Vert
\mathbf{r}_{0}-\mathbf{p}_{m}\right\Vert ^{k}}\\
\bar{s}_{M}^{(k)}  &  =\frac{1}{2M+1}\sum_{m=-M}^{M}\frac{m\Delta_{T}-y_{0}
}{\left\Vert \mathbf{r}_{0}-\mathbf{p}_{m}\right\Vert ^{k+1}}
\end{align*}
where $k\in\mathbb{N}$.

\begin{proposition} \label{prop:quadric}
When $x_{0}=0$ (intended receiver on the $yz$-plane), the identity in (\ref{eq:TaylorSNRsimple}) particularizes to
\begin{equation}
\mathsf{SNR}(\mathbf{r})=\mathsf{SNR}(\mathbf{r}_{0})\left[  1-\left(
2\mathbf{\Delta}_{r}^{T}\mathfrak{m}_{M}+\mathbf{\Delta}_{r}^{T}
\mathcal{M}_{M}\mathbf{\Delta}_{r}\right)  \right]  +\epsilon_{M}
\label{eq:SNRtaylorx0=0}
\end{equation}
where the column vector $\mathfrak{m}_{M}$ and the matrix $\mathcal{M}_{M}$
are defined as follows. The column vector $\mathfrak{m}_{M}$ only has non-zero
entries in the second and third position, so that we can write $\mathfrak{m}_M = [0,(\mathfrak{m}^{(2)}_M)^T]^T$ where
% $\mathfrak{m}
% _{M}=\left[
% \begin{array}
% [c]{cc}
% 0 & \mathfrak{m}_{M}^{(2)}
% \end{array}
% \right]  ^{T}$, where
\[
\mathfrak{m}_{M}^{(2)}=\frac{1}{s_{M}^{(2)}}\left[
\begin{array}
[c]{cc}
-\bar{s}_{M}^{(3)} & z_{0}s_{M}^{(4)}
\end{array}
\right]  ^{T}.
\]
Regarding $\mathcal{M}_{M}$, it is a block-diagonal matrix with diagonals
given by a scalar $\gamma_{M}^{(1)}$ and a $2\times2$ real-valued
matrix $\mathcal{M}_{M}^{(2)}$, that is
\begin{equation}
\mathcal{M}=\left[
\begin{array}
[c]{cc}
\gamma_{M}^{(1)} & \mathbf{0}\\
\mathbf{0} & \mathcal{M}_{M}^{(2)}
\end{array}
\right]  .\label{eq:defcalM}
\end{equation}
The scalar entry takes the form
\begin{equation}
\gamma_{M}^{(1)}=\frac{1}{\left(  s_{M}^{(2)}\right)  ^{2}}\left[
3s_{M}^{(2)}s_{M}^{(4)}-\left(  \bar{s}_{M}^{(3)}\right)  ^{2}-z_{0}
^{2}\left(  s_{M}^{(4)}\right)  ^{2}\right]  \label{eq:defgamma1}
\end{equation}
and is always non-negative, i.e. $\gamma_{M}^{(1)}\geq0$. As for the
$2\times2$ matrix $\mathcal{M}_{M}^{(2)}$, it can be written as $\mathcal{M}
_{M}^{(2)}=\mathcal{A}_{M}+(2\pi/\lambda)^{2}\mathcal{B}_{M}$ with
$\mathcal{A}_{M}$ and $\mathcal{B}_{M}$ being respectively defined as in (\ref{eq:def_matrixAM})-(\ref{eq:def_matrixBM}) at the top of the next page. 
\end{proposition}
\begin{figure*}
    \begin{align}
\mathcal{A}_{M} &  =\frac{1}{\left(  s_{M}^{(2)}\right)  ^{2}}\left[
\begin{array}
[c]{cc}
\left(  3z_{0}^{2}s_{M}^{(6)}-2s_{M}^{(4)}\right)  s_{M}^{(2)}-\left(  \bar
{s}_{M}^{(3)}\right)  ^{2} & \left(  s_{M}^{(4)}\bar{s}_{M}^{(3)}+3s_{M}
^{(2)}\bar{s}_{M}^{(5)}\right)  z_{0}\\
\left(  s_{M}^{(4)}\bar{s}_{M}^{(3)}+3\bar{s}_{M}^{(5)}s_{M}^{(2)}\right)
z_{0} & \left(  s_{M}^{(4)}-3s_{M}^{(6)}z_{0}^{2}\right)  s_{M}^{(2)}-\left(
s_{M}^{(4)}z_{0}\right)  ^{2}
\end{array}
\right]  \label{eq:def_matrixAM} \\
\mathcal{B}_{M} &  =\frac{1}{\left(  s_{M}^{(2)}\right)  ^{2}}\left[
\begin{array}
[c]{cc}
\left(  s_{M}^{(2)}\right)  ^{2}-z_{0}^{2}s_{M}^{(2)}s_{M}^{(4)}-\left(
\bar{s}_{M}^{(2)}\right)  ^{2} & \left(  \bar{s}_{M}^{(2)}s_{M}^{(3)}-\bar
{s}_{M}^{(3)}s_{M}^{(2)}\right)  z_{0}\\
\left(  s_{M}^{(3)}\bar{s}_{M}^{(2)}-\bar{s}_{M}^{(3)}s_{M}^{(2)}\right)
z_{0} & z_{0}^{2}\left(  s_{M}^{(4)}s_{M}^{(2)}-\left(  s_{M}^{(3)}\right)
^{2}\right)
\end{array}
\right]  \label{eq:def_matrixBM}
\end{align}
\hrule
\end{figure*}
\begin{proof}
See Appendix~\ref{sec:AppendixQuadric}.
\end{proof}

If we disregard the high order error terms, which decay as $O(\left\Vert
\mathbf{\Delta}_{r}\right\Vert ^{3})$, the set of points for which the
$\mathsf{SNR}(\mathbf{r})=\kappa\mathsf{SNR}(\mathbf{r}_{0})$ where $\kappa
\in\left(  0,1\right)  $ is asymptotically described by the equation
\begin{equation} \label{eq:quarticFiniteM}
1-\kappa=2\mathbf{\Delta}_{r}^{T}\mathfrak{m}_{M}+\mathbf{\Delta}_{r}
^{T}\mathcal{M}_{M}\mathbf{\Delta}_{r}
\end{equation}
or alternatively as
\[
\left[
\begin{array}
[c]{cc}
\mathbf{\Delta}_{r} & 1
\end{array}
\right]  \left[
\begin{array}
[c]{cc}
\mathcal{M}_{M} & \mathfrak{m}_{M}\\
\mathfrak{m}_{M}^{T} & -(1-\kappa)
\end{array}
\right]  \left[
\begin{array}
[c]{c}
\mathbf{\Delta}_{r}\\
1
\end{array}
\right]  =0
\]
which is the equation of a quadric surface \cite{Coolidge68}. Disregarding all degenerated situations,
this quadric surface can either be an ellipsoid or an hyperboloid (of one or two sheets). The beamfocusing feasibility region will therefore be associated to the region of points in space for which this quadric is an ellipsoid, which is the only case where the enclosed region has with finite volume.
This situation occurs if and only if the three eigenvalues of $\mathcal{M}
_{M}$ are positive. 

Given the structure of $\mathfrak{m}_{M}$ and
$\mathcal{M}_{M}$, we see that the ellipsoid will always be symmetric with respect to the
$yz$-plane (where the receiver is assumed to be located). The length of the three semi-axes, which we denote as $l_{M}
^{(1)}\left(  \mathbf{r}_{0}\right)  $ (along the $x$-axis), $l_{M}
^{(2)}\left(  \mathbf{r}_{0}\right)  $ and $l_{M}^{(3)}\left(  \mathbf{r}
_{0}\right)  $ (on the $yz$-plane) can be expressed in closed form. Indeed, if
we denote as $\gamma_{M}^{(1)}\geq\gamma_{M}^{(2)}\geq\gamma_{M}^{(3)}>0$ the
three eigenvalues of $\mathcal{M}$ (so that $\gamma_{M}^{(1)}$\ is as defined
in (\ref{eq:defgamma1}) and $\gamma_{M}^{(2)},\gamma_{M}^{(3)}$ are the
eigenvalues of $\mathcal{M}_{M}^{(2)}$), we can express
\begin{equation}
l_{M}^{(k)}\left(  \mathbf{r}_{0}\right)  =\sqrt{\frac{\mu_{M}}{\gamma
_{M}^{(k)}}}\label{eq:deflM(k)}
\end{equation}
where
\begin{align*}
\mu_{M} & =\frac{-1}{\det\mathcal{M}_{M}}\det\left[
\begin{array}
[c]{cc}
\mathcal{M}_{M} & \mathfrak{m}_{M}\\
\mathfrak{m}_{M}^{T} & \kappa-1
\end{array}
\right]  \\ & =1-\kappa+\left(  \mathfrak{m}_{M}^{(2)}\right)  ^{T}\left(
\mathcal{M}_{M}^{(2)}\right)  ^{-1}\mathfrak{m}_{M}^{(2)}.
\end{align*}
In particular, the volume enclosed by the ellipsoid where the achieved SNR is
at least $\kappa\mathsf{SNR}(\mathbf{r}_{0})$ can be expressed as
\[
\mathcal{V}_{M}\left(  \mathbf{r}_{0}\right)  =\frac{4\pi}{3}l_{M}
^{(1)}\left(  \mathbf{r}_{0}\right)  l_{M}^{(2)}\left(  \mathbf{r}_{0}\right)
l_{M}^{(3)}\left(  \mathbf{r}_{0}\right)  =\frac{4\pi}{3}\sqrt{\frac{\mu
_{M}^{3}}{\gamma_{M}^{(1)}\gamma_{M}^{(2)}\gamma_{M}^{(3)}}}
\]
which effectively goes to infinity as $\mathcal{M}_{M}^{(2)}$ becomes
singular. The center of the ellipsoid is located at the point $\mathbf{r}
_{0}-\mathcal{M}_{M}^{-1}\mathfrak{m}_{M}$ \ of the three dimensional space
(equivalently at the point $[y_{0},z_{0}]^{T}-(\mathcal{M}_{M}^{(2)}
)^{-1}\mathfrak{m}_{M}^{(2)}$ of the $yz$-plane).

All these parameters allows us to establish the beamfocusing feasibility region 
(that is the region where the quadric is an ellipsoid) and the local effectiveness of 
beamfocusing operation (given by the volume of the corresponding ellipsoid for any fixed $\kappa$). From the analytical perspective, though, it is difficult to gain much insight from the above
expressions, since the different terms need to be evaluated for each fixed $M$
and $\Delta_{T}$. In the next section we consider the holographic regime
whereby the number of elements of the ULA increases to infinity
($M\rightarrow\infty$) while the distance between consecutive elements
converges to zero at the same rate ($\Delta_{T}\rightarrow0$)\ so that the
total dimension of the aperture converges to $(2M+1)\Delta_{T}\rightarrow2L>0$. 
It can be seen that the approximation is quite accurate even in situations
where the interelement separation is relatively high, as will be shown in the
numerical evaluation study below.

\section{The holographic regime} \label{sec:holographic}

The above analysis can be simplified by considering the holographic regime,
according to which we allow $M\rightarrow\infty$ and $\Delta_{T}\rightarrow0$
so that $M\Delta_{T}\rightarrow L$, $0<L<\infty$. In these circumstances, one
can easily see that $s_{M}^{(k)}\rightarrow\chi_{k}$ and $\bar{s}_{M}^{(k)}\rightarrow\bar{\chi}_{k}$, where the quantities ${\chi}_{k}$ and $\bar{\chi}_{k}$ are defined as follows. On the one hand, for $k\in\mathbb{N}$ the quantities $\bar{\chi}_{k}$ are given by 
% \[
% \bar{\chi}_{k} = \frac{1}{2L(k-1)}\left[  \frac{1}{\left((y + y_0)^2+ z_0^2\right)^{\frac{k-1}{2}}}  -  \frac{1}{\left((y - y_0)^2+ z_0^2\right)^{\frac{k-1}{2}}}  \right]
% \]
\begin{multline*}
\bar{\chi}_{k} = \frac{1}{2L(k-1)}\big[ \left((L + y_0)^2+ z_0^2\right)^{-(k-1)/2} \\ - \left((L - y_0)^2+ z_0^2\right)^{-(k-1)/2} \big]    
\end{multline*}
% \begin{align*}
% \bar{\chi}_{2}  &  =\frac{-1}{2L}\left[  \frac{1}{\left(  \left(
% y_{0}-L\right)  ^{2}+z_{0}^{2}\right)  ^{1/2}}-\frac{1}{\left(  \left(
% y_{0}+L\right)  ^{2}+z_{0}^{2}\right)  ^{1/2}}\right] \\
% \bar{\chi}_{3}  &  =\frac{-1}{4L}\left[  \frac{1}{\left(  y_{0}-L\right)
% ^{2}+z_{0}^{2}}-\frac{1}{\left(  y_{0}+L\right)  ^{2}+z_{0}^{2}}\right] \\
% \bar{\chi}_{4}  &  =\frac{-1}{6L}\left[  \frac{1}{\left(  \left(
% y_{0}-L\right)  ^{2}+z_{0}^{2}\right)  ^{3/2}}-\frac{1}{\left(  \left(
% L+y_{0}\right)  ^{2}+z_{0}^{2}\right)  ^{3/2}}\right] \\
% \bar{\chi}_{5}  &  =\frac{-1}{8L}\left[  \frac{1}{\left(  \left(
% y_{0}-L\right)  ^{2}+z_{0}^{2}\right)  ^{2}}-\frac{1}{\left(  \left(
% y_{0}+L\right)  ^{2}+z_{0}^{2}\right)  ^{2}}\right]
% \end{align*}
On the other hand, the quantities $\chi_{k}$ for $k=2,3,4,6$ are defined as  \cite{agustin24twc}:
\begin{align*}
\chi_{2}  &  =\frac{1}{2Lz_{0}}\left[  \arctan\frac{L-y_{0}}{z_{0}}
+\arctan\frac{L+y_{0}}{z_{0}}\right] \\
\chi_{3}  &  =\frac{1}{2Lz_{0}^{2}}\left[  \frac{L-y_{0}}{\sqrt{{\left(
L-y_{0}\right)  ^{2}+z_{0}^{2}}}}+\frac{L+y_{0}}{\sqrt{ \left(
y_{0}+L\right)  ^{2}+z_{0}^{2}}}\right] \\
\chi_{4}  &  =\frac{1}{4Lz_{0}^{2}}\left[  \frac{L-y_{0}}{  \left(
L-y_{0}\right)  ^{2}+z_{0}^{2} }+\frac{L+y_{0}}{  \left(
L+y_{0}\right)  ^{2}+z_{0}^{2} }\right]  +\frac{\chi_{2}}{2z_{0}^{2}
}\\
\chi_{6}  &  =-\frac{1}{8Lz_{0}^{4}}\left[  \frac{\left(
L-y_{0}\right)  ^{3}}{\left(  \left(  L-y_{0}\right)  ^{2}+z_{0}^{2}\right)
^{2}}+\frac{\left(  L+y_{0}\right)  ^{3}}{\left(  \left(  L+y_{0}\right)
^{2}+z_{0}^{2}\right)  ^{2}}\right] \\ & +\frac{5\chi_{4}}{4z_{0}^{2}}-\frac
{\chi_{2}}{4z_{0}^{4}}.
\end{align*}
This implies that $\mathfrak{m}_{M}$ in Proposition~\ref{prop:quadric} will converge to $\mathfrak{m}$ given by
$\mathfrak{m}=\left[  0,\mathfrak{m}_{2}\right]  ^{T}$ where
\begin{equation}
\mathfrak{m}_{2}=\frac{1}{\chi_{2}}\left[
\begin{array}
[c]{c}
-\bar{\chi}_{3}\\
\chi_{4}z_{0}
\end{array}
\right]. \label{eq;defm2}
\end{equation}
On the other hand, $\gamma_{M}^{(1)}$ will converge to
\begin{equation}
\gamma_{1}=\frac{3\chi_{2}\chi_{4}-\bar{\chi}_{3}^{2}-z_{0}^{2}\chi_{4}^{2}
}{\chi_{2}^{2}}    \label{eq:gamma1}
\end{equation}
and $\mathcal{M}_{M}^{(2)}$ will converge to $\mathcal{M}_{2}$ given by
\begin{align}
& \mathcal{M}_{2} = \label{eq;defM2}\\
&  \frac{1}{\chi_{2}^{2}}\left[
\begin{array}
[c]{cc}
(3\chi_{6}z_{0}^{2}-2\chi_{4})\chi_{2}-\bar{\chi}_{3}^{2} & \left(
3\chi_{2}\bar{\chi}_{5}+\chi_{4}\bar{\chi}_{3}\right)  z_{0}\\
\left(  3\chi_{2}\bar{\chi}_{5}+\chi_{4}\bar{\chi}_{3}\right)  z_{0} &
\chi_{2}\chi_{4}-z_{0}^{2}\left(  \chi_{4}^{2}+3\chi_{2}\chi_{6}\right)
\end{array}
\right] \nonumber \\
&  +\left(  \frac{2\pi}{\lambda}\right)  ^{2}\frac{1}{\chi_{2}^{2}}\left[
\begin{array}
[c]{cc}
\chi_{2}^{2}-z_{0}^{2}\chi_{2}\chi_{4}-\bar{\chi}_{2}^{2} & \left(  \bar{\chi
}_{2}\chi_{3}-\chi_{2}\bar{\chi}_{3}\right)  z_{0}\\
\left(  \bar{\chi}_{2}\chi_{3}-\chi_{2}\bar{\chi}_{3}\right)  z_{0} & \left(
\chi_{2}\chi_{4}-\chi_{3}^{2}\right)  z_{0}^{2}
\end{array}
\right]. \nonumber
\end{align}
The first term on the right hand side of (\ref{eq:SNRtaylorx0=0}) will also
converge in the holographic regime to a quantity obtained by replacing
$\mathfrak{m}_{M}$ and $\mathcal{M}_{M}$ with $\mathfrak{m}$ and $\mathcal{M}$
respectively. Observe also that, since $\left\vert \epsilon_{M}\right\vert
\leq K\left\Vert \mathbf{\Delta}_{r}\right\Vert ^{3}$ for a certain nice
constant $K$, we will also have $\lim\sup_{M\rightarrow\infty}\left\vert
\epsilon_{M}\right\vert \leq K\left\Vert \mathbf{\Delta}_{r}\right\Vert ^{3}$
and the error term will also be of order $O(\left\Vert \mathbf{\Delta}
_{r}\right\Vert ^{3})$ in the holographic regime.

Noting that $\gamma_{1}>0$ one can conclude that the region where
beamfocusing is feasible can be identified with the region of space where the
two eigenvalues of $\mathcal{M}_{2}$ are positive, i.e. $\mathcal{M}_{2}>0$.
Here again, this corresponds to the corresponding quadric surface being an
ellipsoid symmetric with respect to the $yz$-plane. The three semiaxes of the
ellipsoid will now be given by the limits of $l_{M}^{(k)}$, $k=1,2,3$, in
(\ref{eq:deflM(k)}). In other words, we will have $l_{M}^{(k)}\rightarrow
l_{k}=\sqrt{\mu/\gamma_{k}}$ where $\gamma_{k}$, $k=1,2,3$, are the
eigenvalues of $\mathcal{M}$ and where $\mu=1-\kappa+\mathfrak{m}_{2}
^{T}\mathcal{M}_{2}^{-1}\mathfrak{m}_{2}$ with $\mathfrak{m}_{2}$ and
$\mathcal{M}_{2}$ as defined in (\ref{eq;defm2})-(\ref{eq;defM2}). The
ellipsoid itself will be centered at the point $[y_{0},z_{0}]^{T}
-\mathcal{M}_{2}^{-1}\mathfrak{m}_{2}$ of the $yz$-plane.

\begin{remark} \label{rem:decayment_xi}
It can be seen that $D^{k}\chi_{k}$ and $D^{k}\bar{\chi}_{k}$ are functions of
$\varrho=L/D$ and $\theta$ and no other variable, where we recall that $\theta$ denotes the elevation angle of the intended receiver and $D$ is the distance between the intended receiver and the center of the ULA (see Fig.~\ref{fig:Scenario}). 
This implies that $D^{2}\gamma_{1}$ and $D\mathfrak{m}_{2}$ are also a function of $\varrho$ and $\theta$\ only,
whereas $D^{2}\mathcal{M}_{2}$ can be expressed as a function of $\varrho$,
$\theta$ and $L/\lambda$. In particular, we observe that the eigenvalues of
$D^{2}\mathcal{M}_{2}$ are also functions of these three variables only. We can
therefore establish the region where beamfocusing is possible as the region
where the three eigenvalues of this matrix are positive, which can easily be
evaluated for any fixed value of $L/\lambda$. 
\end{remark}
Indeed, by considering a range of values of the elevation $\theta$ in the set $(-\pi /2,\pi/2)$ (see Fig.~\ref{fig:Scenario}), one can evaluate the minimum eigenvalue of $D^{2}\mathcal{M}_{2}$
as a function of $\varrho$ and establish the values of this quantity for which
the minimum eigenvalue is positive. The region where beamfocusing is feasible
will easily be established for each value of $L/\lambda$ under consideration.
This is illustrated in Fig.~\ref{fig:beamfocusingRegion}, were we represent the beamfocusing feasibility
regions for different values of this parameter. The boundary of the feasibility region is numerically obtained by evaluating the minimum eigenvalue of $D^2\mathcal{M}_2$ as a function of $\varrho = L/D$ and finding the values of $\varrho$ where this eigenvalue becomes zero. In general terms, it can be
observed that depending on the value of the elevation $\theta$, the
feasibility region is either empty or consists of a single interval of values
in $\varrho$.    
\begin{figure}
    \centering
    \includegraphics[width=\linewidth]{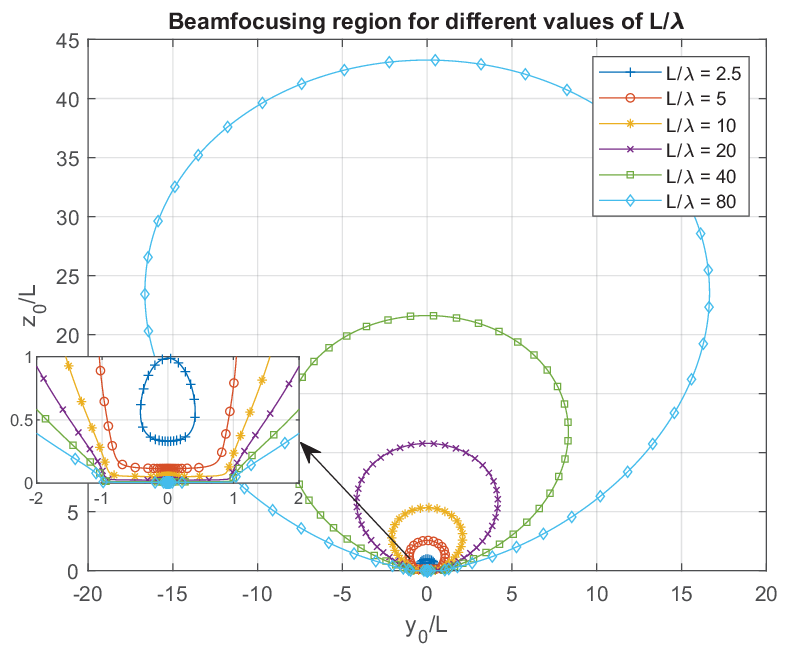}
    \caption{Spatial regions where beamfocusing is feasible for different values of the array size relative to the wavelength in the holographic regime ($L/\lambda$).}
    \label{fig:beamfocusingRegion}
\end{figure}

The quotient between the semiaxes of the ellipsoid and the
distance $D$ (that is $l_{k}/D$) will also be a function of $\varrho,$
$\theta$\ and $L/\lambda$ only. Furthermore, if we denote by $\mathcal{V}
\left(  \mathbf{r}_{0}\right)  $ the volume of the ellipsoid in the
holographic regime, that is
\[
\mathcal{V}\left(  \mathbf{r}_{0}\right)  =\frac{4\pi}{3}l_{1}l_{2}l_{3}
=\frac{4\pi}{3}\sqrt{\frac{\mu^{3}}{\gamma_{1}\gamma_{2}\gamma_{3}}}
\]
we have that $\mathcal{V}\left(  \mathbf{r}_{0}\right)  /D^{3}$ is a function
of $\varrho,$ $\theta$\ and $L/\lambda$ only. This quantity provides a
quantitative measure of the efficiency of beamfocusing within the feasibility
region, see further Fig.\ref{fig:volEllipsoid}.
\begin{figure}
    \centering
    \includegraphics[width=\linewidth]{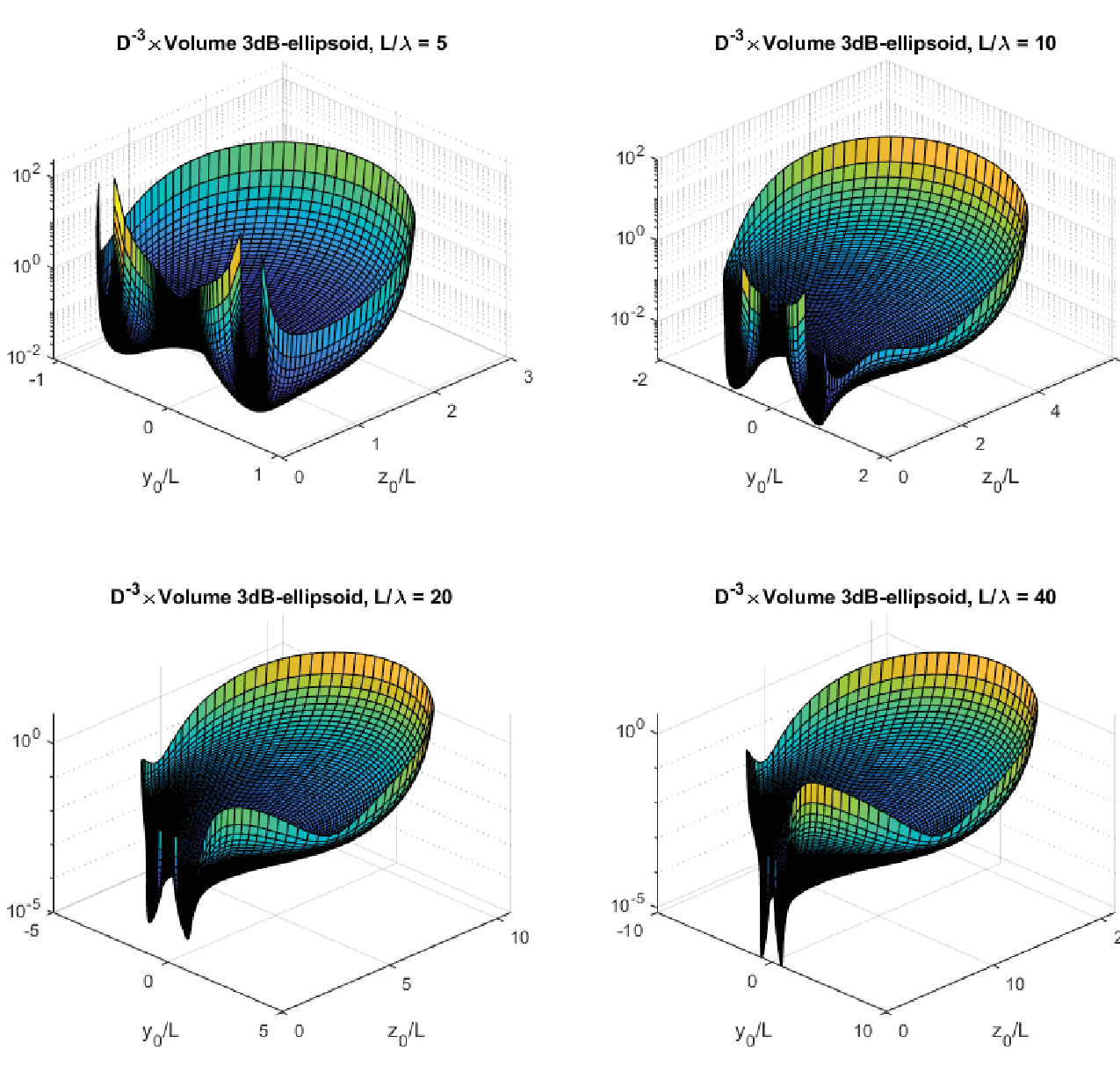}
    \caption{Volume of the asymptotical ellipsoid that characterizes the 3dB loss with respect to the maximum SNR ($\kappa = 0.5$) for different values of $L/\lambda$.}
    \label{fig:volEllipsoid}
\end{figure}

The above description in the holographic regime provides an analytical
framework independent of $M$ that can be used to establish the regions of
beamfocusing and the corresponding $\kappa$-efficiency as a function of
position of the intended receiver. In the following two subsections, we try to
provide some more insight by focusing on two special cases. We will first
consider the case where the intended receiver is located at the broadside of
the array, so that $y_{0}=0$. Next, we will provide an approximation of the
above quantities for reasonably high values of $D/L\,$, where we recall that $D=\sqrt
{y_{0}^{2}+z_{0}^{2}}$ is the distance between the intended receiver and the
center of the ULA and $L$ is the one-sided dimension of the aperture.

\subsection{Case $y_{0}=0$ (intended receiver on the broadside)}

When $y_{0}=0$ we can simplify the above expressions by noting that $\bar
{\chi}_{k}=0$ for all $k$, whereas $\chi_{2} =(Lz_{0})^{-1}\arctan(L/z_0)$ and similar simplifications are obtained for $\chi_{3},\chi_{4},\chi_{6}$.
% \begin{align*}
% \chi_{3} &  =\frac{1}{z_{0}^{2}}\frac{1}{\left(  L^{2}+z_{0}^{2}\right)
% ^{1/2}}\\
% \chi_{4} &  =\frac{1}{2z_{0}^{2}}\frac{1}{\left(  L^{2}+z_{0}^{2}\right)
% }+\frac{\chi_{2}}{2z_{0}^{2}}\\
% \chi_{6} &  =\frac{1}{8z_{0}^{4}}\frac{3L^{2}+5z_{0}^{2}}{\left(  L^{2}
% +z_{0}^{2}\right)  ^{2}}+\frac{3}{8z_{0}^{4}}\chi_{2}
% \end{align*}
From this, we find that $\gamma_{1}$ particularizes to
\[
{\gamma}_{1}=\frac{1}{4z_{0}^{2}\chi_{2}^{2}}\left(  5\chi_{2}-\frac
{1}{\left(  L^{2}+z_{0}^{2}\right)  }\right)  \left(  \chi_{2}+\frac
{1}{\left(  L^{2}+z_{0}^{2}\right)  }\right)  .
\]
On the other hand, $\mathcal{M}_{2}$ becomes diagonal, with diagonal entries
that we will denote $\gamma_{2}$, $\gamma_{3}$ given by
\begin{align*}
\gamma_{2} &  =\frac{1}{8\chi_{2}z_{0}^{2}}\left(  \frac
{L^{2}+7z_{0}^{2}}{\left(  L^{2}+z_{0}^{2}\right)  ^{2}}+\chi_{2}\right)
\\ &+\frac{1}{2}\left(  \frac{2\pi}{\lambda}\right)  ^{2}\left(  1-\frac{1}
{\chi_{2}\left(  L^{2}+z_{0}^{2}\right)  }\right)  \\
\gamma_{3} &  =-\frac{1}{\chi_{2}^{2}}\frac{1}{8z_{0}^{2}}\left(  7\chi
_{2}^{2}+3\frac{3L^{2}+5z_{0}^{2}}{\left(  L^{2}+z_{0}^{2}\right)  ^{2}}
\chi_{2}+\frac{2}{\left(  L^{2}+z_{0}^{2}\right)  ^{2}}\right)  \\
&  +\left(  \frac{2\pi}{\lambda}\right)  ^{2}\frac{1}{2\chi_{2}^{2}}\left(
\chi_{2}^{2}+\frac{\chi_{2}}{\left(  L^{2}+z_{0}^{2}\right)  }-\frac{1}
{z_{0}^{2}}\frac{2}{\left(  L^{2}+z_{0}^{2}\right)  }\right)  .
\end{align*}
The first two eigenvalues of $\mathcal{M}$ (i.e. $\gamma_{1}$, $\gamma_{2}$) are positive, since
\[
\chi_{2}\left(  L^{2}+z_{0}^{2}\right)  =\left(  1+\frac{z_{0}^{2}}{L^{2}
}\right)  \frac{L}{z_{0}}\arctan\frac{L}{z_{0}}>1
\]
so that beamfocusing will be feasible whenever $\gamma_{3}>0$. Defining $\varphi
_{2}\left(  \varrho\right)  =\varrho^{-1}\arctan\varrho$ where $\varrho=L/D$,
one can characterize the beamfocusing feasibility segment as the set of points for which $\gamma_{3}>0,$ a
condition that can be expressed as
\begin{equation}
\frac{\varrho}{2\pi }\sqrt{\frac{7\left(
1+\varrho^{2}\right)  ^{2}\varphi_{2}^{2}\left(  \varrho\right)  +3\left(
5+3\varrho^{2}\right)  \varphi_{2}\left(  \varrho\right)  +2}{ 4\left(  1+\varrho^{2}\right) \left(\varphi_{2}
^{2}\left(  \varrho\right)  \left(  1+\varrho^{2}\right)  +\varphi_{2}\left(
\varrho\right)  -2\right)}} \leq \frac{ L}{\lambda}
.\label{eq:beamfocusBroadside}
\end{equation}
The left hand side is a continous function of $\varrho=L/D$ that presents a
single inflection point (local minima) at $\varrho=1.72776$ where it takes the
value $2.2048$. This is illustrated in  Fig.~\ref{fig:beamfocusingBroadside}, where we represent the left hand side of (\ref{eq:beamfocusBroadside}) as a function of $\varrho^{-1} = D/L$. The region of
feasibility of beamfocusing can be described as the region between the two
roots of (\ref{eq:beamfocusBroadside}), which can be obtained by considering the intersection of the blue solid line in Fig.~\ref{fig:beamfocusingBroadside} with a horizontal line at $L/\lambda$. The inflection point implies that beamfocusing along the
broadside of the array is only possible when $L>2.2048\lambda$. 
\begin{figure}
    \centering
    \includegraphics[width=\linewidth]{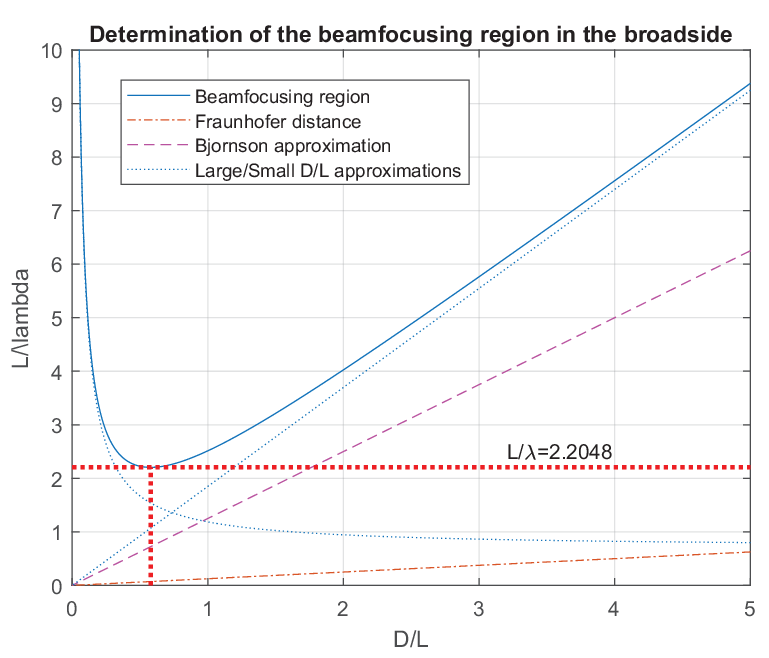}
    \caption{Representation of the function on the left hand side of (\ref{eq:beamfocusBroadside}) as a function of $\varrho^{-1} = D/L$. The feasibility segment for beamfocusing on the broadside is obtained by selecting the interval where this curve is lower than $L/\lambda$. Additionally, we represent the two lines that determine the Fraunhofer distance by the same procedure (red dash-dotted line) as well as the beamfocusing radius approximation in \cite[Theorem 1]{Bjornson2021primer}, consisting of one tenth of the Fraunhofer distance (magenta dashed line). 
    %It can be observed that the feasibility region is empty when $L/\lambda<2.2048$.
    }
    \label{fig:beamfocusingBroadside}
\end{figure}
We can gain some insight into the beamfocusing feasibility region by considering a large/small-$\varrho$ approximation of the left hand side of (\ref{eq:beamfocusBroadside}). Indeed, for large values of $\varrho$ the left hand side of (\ref{eq:beamfocusBroadside}) can be expanded as
\[
\frac{3\sqrt{15}}{2\pi \varrho} +o(1)
\]
whereas for small values of $\varrho$ it can be approximated as 
\[
\frac{1}{4}\sqrt{\frac{7}{\pi^2-8}}\left(\varrho + \frac{\pi^2+20}{14\pi\left(\pi^2-8\right)}\right) + o(1).
\]
These two approximations are shown in Fig.~\ref{fig:beamfocusingBroadside} in blue dotted lines. 
% For large values of
% $\varrho$ the 
% argument inside the square root of (\ref{eq:beamfocusBroadside}) can be expanded as
% \[
% \frac{7\pi^{2}}{4\left(  \pi^{2}-8\right)  }\varrho^{2}+\frac{\pi(20+\pi^{2}
% )}{\left(  \pi^{2}-8\right)  ^{2}}\varrho-\frac{4\pi^{2}\left(  3\pi
% ^{2}-38\right)  }{(\pi^{2}-8)^{3}}+O(\varrho^{-1}).
% \]
Using these approximations, we can establish that for large values of $D/L$, the maximum distance that guarantees the feasibility of beamfocusing can be approximated as $D_{\max} \approx 2\pi L^2/(3\lambda\sqrt{15}) = \pi/(12\sqrt{15})D_\mathrm{Fraun}$, where $D_\mathrm{Fraun} = 8L^2/\lambda$ is the Fraunhofer distance. For small values of  $D/L$ the minimum distance can be approximated as
\[
D_{\min} \approx {L} \left( \frac{4L}{\lambda}\sqrt{\frac{\pi^2-8}{7}} - \frac{\pi^2+20}{14\pi(\pi^2-8)} \right)^{-1}
\]
% $D_\min/L \approx 4L/\lambda\sqrt{(\pi^2-8)/7}-(\pi^2+20)/(14\pi(\pi^2-8))$, 
which becomes proportional to $\lambda$ for large apertures. This is quite often within the reactive region of the electromagnetic field (i.e. below the Fresnel distance). 
% This means that for small values of $D/L$, the minium distance to guarantee
% beamfocusing can be approximated by (\ref{eq:minimumFeasibleDistance}) at the top of page~\pageref{eq:minimumFeasibleDistance}, which decays for large $L/\lambda$ as $O((L/\lambda)^{-1}).$ 
%Consequently, the larger the ULA dimensions, the closer to the array can beamfocusing become feasible (see the zoomed region in Fig.~\ref{fig:beamfocusingRegion}). 
% \begin{figure*} [t]
%     \begin{equation}
%         \left(  \frac{D}{L}\right)  _{\min}=\frac{\left(  \pi^{2}-8\right)  }{8\pi
% }\frac{20+\pi^{2}+\sqrt{\left(  85\pi^{4}-1024\pi^{2}+400\right)  +28\pi
% ^{2}\left(  \pi^{2}-8\right)  ^{3}\left(  L/\lambda\right)  ^{2}}}{(3\pi
% ^{2}-38)+\left(  \pi^{2}-8\right)  ^{3}\left(  L/\lambda\right)  ^{2}} \label{eq:minimumFeasibleDistance}
%     \end{equation}
%     \hrule
% \end{figure*}

\subsection{Approximation for high $D/L$}

We now turn to the more interesting case where the receiver is not necessarily located in the broadside of the array, so that the angle $\theta$ in Fig.~\ref{fig:Scenario} can take any value in the set $(-\pi/2,\pi/2)$. In particular, we will try to gain some insights into the beamfocusing achievability region and the actual shape of the corresponding achievability ellipses when the distance $D$ is large compared to the ULA aperture side length $L$.

As pointed out in Remark~\ref{rem:decayment_xi}, both $D^{k}\chi_{k}$ and $D^{k}\bar{\chi}_{k}$ can be expressed as analytic functions of
$\varrho = L/D$ and $\theta$ only. In particular, in order to gain some further insight into the behavior of the above description, one can consider the Taylor series approximation of all these functions around $\varrho =0$. 
The Taylor series of $D^k\chi_k$ for $k=2,4,6$ and $D^k\bar{\chi}_k$ for $k=5$ around $\varrho=L/D=0$ were presented in \cite[Appendix A]{mestre25ojsp}. Similarly, one can see that 
\begin{align*}
D^{3}\chi_{3}  &  =1+\left(  2-{5/2\cos^{2}\theta}\right)  \varrho
^{2}+O\left(  \varrho^{4}\right) \\
D^{2}\bar{\chi}_{2}  &  =-\sin\theta\left[  1+\left(  1-5/2\cos
^{2}\theta\right)  \varrho^{2} \right] +O\left(  \varrho^{4}\right) \\
D^{3}\bar{\chi}_{3}  &  =-\sin\theta\left[  1+\left(  2-4\cos^{2}
\theta\right)  \varrho^{2}\right]  +O\left(  \varrho^{4}\right).
\end{align*}
% we can readily
% see that these quantities behave for low $\varrho$ as
% \begin{align*}
% D^{2}\chi_{2}  &  =1+\left(  1-\frac{4\cos^{2}\theta}{3}\right)  \varrho
% ^{2}+\frac{5-20\cos^{2}\theta+16\cos^{4}\theta}{5}\varrho^{4}+O\left(
% \varrho^{6}\right) \\
% D^{3}\chi_{3}  &  =1+\left(  2-\frac{5\cos^{2}\theta}{2}\right)  \varrho
% ^{2}+\frac{3}{8}\left(  8-28\cos^{2}\theta+21\cos^{4}\theta\right)
% \varrho^{4}+O\left(  \varrho^{6}\right) \\
% D^{4}\chi_{4}  &  =1+\left(  \frac{10}{3}-4\cos^{2}\theta\right)  \varrho
% ^{2}+\left(  7-\frac{112}{5}\cos^{2}\theta+16\cos^{4}\theta\right)
% \varrho^{4}+O\left(  \varrho^{6}\right) \\
% D^{6}\chi_{6}  &  =1+\left(  7-8\cos^{2}\theta\right)  \varrho^{2}+O\left(
% \varrho^{4}\right)
% \end{align*}
% Regarding the quantities with bar, we can write
% \begin{align*}
% D^{2}\bar{\chi}_{2}  &  =-\sin\theta\left[  1+\left(  1-\frac{5}{2}\cos
% ^{2}\theta\right)  \varrho^{2}+\left(  1-7\cos^{2}\theta+\frac{63}{8}\cos
% ^{4}\theta\right)  \varrho^{4}\right]  +O\left(  \varrho^{6}\right) \\
% D^{3}\bar{\chi}_{3}  &  =-\sin\theta\left[  1+\left(  2-4\cos^{2}
% \theta\right)  \varrho^{2}+\left(  3-16\cos^{2}\theta+16\cos^{4}\theta\right)
% \varrho^{4}\right]  +O\left(  \varrho^{6}\right) \\
% D^{5}\bar{\chi}_{5}  &  =-\sin\theta\left(  1+(5-8\cos^{2}\theta)\varrho
% ^{2}\right)  +O\left(  \varrho^{4}\right)  .
% \end{align*}
A direct application of these series in (\ref{eq:gamma1}) shows that
\begin{align*}
D^{2}\gamma_{1}  &  =\frac{\left(  5\cos^{2}\theta
-1\right)  \left(  \cos^{2}\theta+1\right)}{4\cos^{6}\theta} \\
&  +\frac{(2\cos^{2}\theta-1)(4\cos^{4}\theta-3\cos
^{2}\theta-3)}{6\cos^{8}\theta}\varrho^{2}+O\left(  \varrho^{4}\right).
\end{align*}
Regarding the matrix $\mathcal{M}_{2}$, it can be expanded as in (\ref{eq:expansionMcal}) at the top of the next page, where we have used the short-hand notation $c_\theta =\cos\theta$ and $s_\theta =\sin\theta$.
\begin{figure*}
\begin{multline}
D^{2}\mathcal{M}_{2}    =\left[
\begin{array}
[c]{cc}
-3+4{c^{2}_\theta} & -4{s_\theta}{c_\theta}\\
-4{s_\theta}{c_\theta} & 1-4{c^2_\theta}
\end{array}
\right] +\left(  \frac{2\pi L}{\lambda}\right)  ^{2}\frac{c^{2}_\theta}{3}\left[
\begin{array}
[c]{cc}
c^{2}_\theta & - s_\theta c_\theta\\
-s_\theta c_\theta & s^{2}_\theta
\end{array}
\right] \\   +\frac{1}{3}\left[
\begin{array}
[c]{cc}
-20+92c^{2}_\theta-76c^{4}_\theta & -s_\theta c_\theta \left(
46-76 c^{2}_\theta\right) \label{eq:expansionMcal} \\
-s_\theta c_\theta \left(  46-76 c^{2}_\theta\right)  & 7-76 c^{2}_\theta+76 c^{4}_\theta
\end{array}
\right]  \varrho^{2} +\left(  \frac{2\pi L}{\lambda}\right)  ^{2}\frac{ c^{2}_\theta}
{45}\left[
\begin{array}
[c]{cc}
c^{2}_\theta\left(  61-74 c^{2}_\theta\right(  & -s_\theta c_\theta\left(  45-74 c^{2}_\theta\right) \\
- s_\theta c_\theta\left(  45-74 c^{2}_\theta\right)  & 30-103 c^{2}_\theta+74 c^{4}_\theta
\end{array}
\right]  \varrho^{2}+O(\varrho^{4}) 
\end{multline}
    \hrule
\end{figure*}
The two eigenvalues of this matrix can in turn be asymptotically expressed as 
\begin{align*}
D^{2}\gamma_{2}  &  = 1+\frac{c^{2}_\theta}{3}\left(  \frac{2\pi L}{\lambda
}\right)  ^{2} +  \\ & +\left[  \frac{7-11 c^{2}_\theta}{3}+ \left(  \frac{2\pi L}{\lambda}\right)  ^{2} \frac{30-43 c^{2}_\theta
}{45} c^{2}_\theta\right]
\varrho^{2}+O(\varrho^{4})\\
D^{2}\gamma_{3}  &  =-3+\left[  \frac{-20+27 c^{2}_\theta}{3}+ \left(  \frac{2\pi L}{\lambda}\right)  ^{2} \frac{c^{4}_\theta}
{45} \right]
\varrho^{2}+O(\varrho^{4}).
\end{align*}
If we disregard the terms of order $O(\varrho^{4})$ in the above expansion for $D^{2}\gamma_{3}$, we can approximate the feasibility region of beamfocusing ($\gamma_3>0$) as the set of points such that
\begin{equation} \label{eq:beamfocusHighD}
\left(  \frac{D}{L}\right)  ^{2}<\frac{1}{9}\left[  -20+27\cos^{2}\theta
+\frac{\cos^{4}\theta}{15} \left(  \frac{2\pi L}{\lambda}\right)  ^{2}\right].
\end{equation}
Fig.~\ref{fig:beamfocusHighD} shows a comparison between the actual beamfocusing feasibility region in the holographic regime (blue solid lines) and the asymptotic approximation in (\ref{eq:beamfocusHighD}) (red dashed lines) for  different values of $L/\lambda$. It can be seen that the approximation is quite accurate for $L/\lambda \geq 10$, that is when the total dimension of the ULA is at least $20\lambda$, which corresponds to the situation where the feasibility region is sufficiently large in terms of $D/L$. 
\begin{figure}
    \centering
    \includegraphics[width=\linewidth]{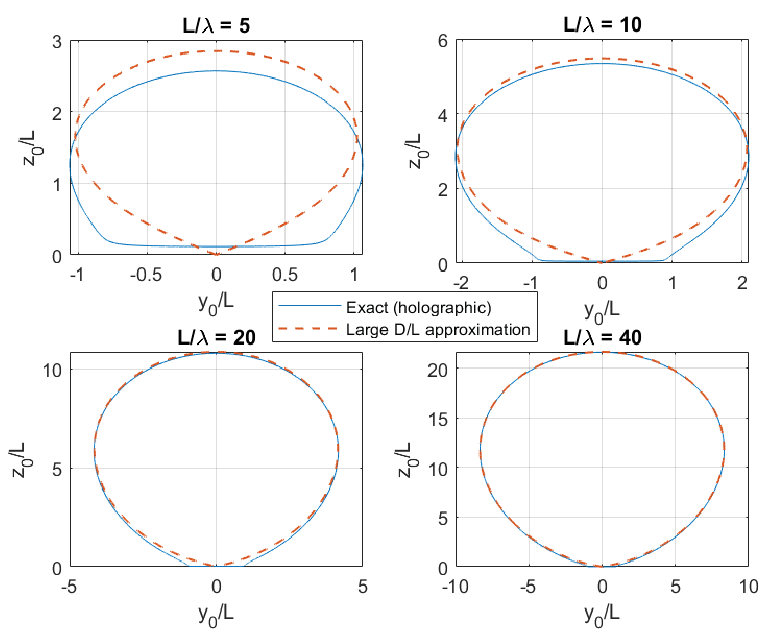}
    \caption{Comparision of the actual feasibility region for beamfocusing and the large $D/L$ approximation for different values of $L/\lambda$.}
    \label{fig:beamfocusHighD}
\end{figure}

Finally, we can characterize the behavior for large $D/L$ of the  three
semiaxes of the corresponding ellipsoid. To see that, we can first find an
expansion $D\mathfrak{m}_{2}$ as a function of $\varrho=L/D$, namely
\[
D\mathfrak{m}_{2}=\left[
\begin{array}
[c]{c}
\sin\theta\left[  1+\left(  1-\frac{8}{3}\cos^{2}\theta\right)  \varrho
^{2}\right] \\
\cos\theta\left[  1+\left(  \frac{7}{3}-\frac{8}{3}\cos^{2}\theta\right)
\varrho^{2}\right]
\end{array}
\right]  +O(\varrho^{4}).
\]
Using this and the expression of $D^{2}\mathcal{M}_{2}$ above one can obtain
\[
\mu=\left(  \frac{2}{3}-\kappa\right)  -\frac{1}{135}\left(  15 c^{2}_\theta -10+\frac{c^{4}_\theta}{3}\left(  \frac{2\pi L}{\lambda}\right)
^{2}\right)  \varrho^{2}+O(\varrho^{4}).
\]
From this, one can obtain the corresponding expansion of the three semi-axes of 
the ellipsoid, which can be used to describe the spatial efficiency of the beamfocusing mechanism. 

\section{Numerical Analysis}

In this section we provide a numerical study of the accuracy of the above asymptotic approximations in a practical setting. In order to justify the holographic approximations, we consider a rather standard situation where the ULA elements are separated by $\Delta_T = \lambda/2$. The transmit array consisted of $2M+1=101$ antennas, so that the total dimension of the array was $2L = 25\lambda$. To generate the actual channel, we assumed a frequency of operation of $3$~GHz, which corresponds to a wavelength equal to $0.1$~m. Fig.~\ref{fig:DeltaLambda2_near} to Fig.~\ref{fig:DeltaLambda2_off} represent the region of the $yz$-plane that achieves a SNR that is at least $3$~dB below the optimum one for different positions of the intended receiver (represented as markers). Apart from the actual region (represented in solid blue lines), we also plot the corresponding conic curves that are obtained by intersecting the corresponding quadric achievability surfaces with the $yz$-plane. More specifically, red dotted lines represent the achievability regions obtained by using the finite-$M$ approximation in Proposition~\ref{prop:quadric}, whereas green dash-dotted lines represent the corresponding holographic representation obtained as $M\rightarrow \infty$ and $\Delta_T\rightarrow 0$ while $M\Delta_T\rightarrow L$, as presented in Section~\ref{sec:holographic}. We also represent, in black dash-dotted line, the holographic approximation to the beamfocusing achievability region. 
\begin{figure}[t]
    \centering
    \includegraphics[width=\linewidth]{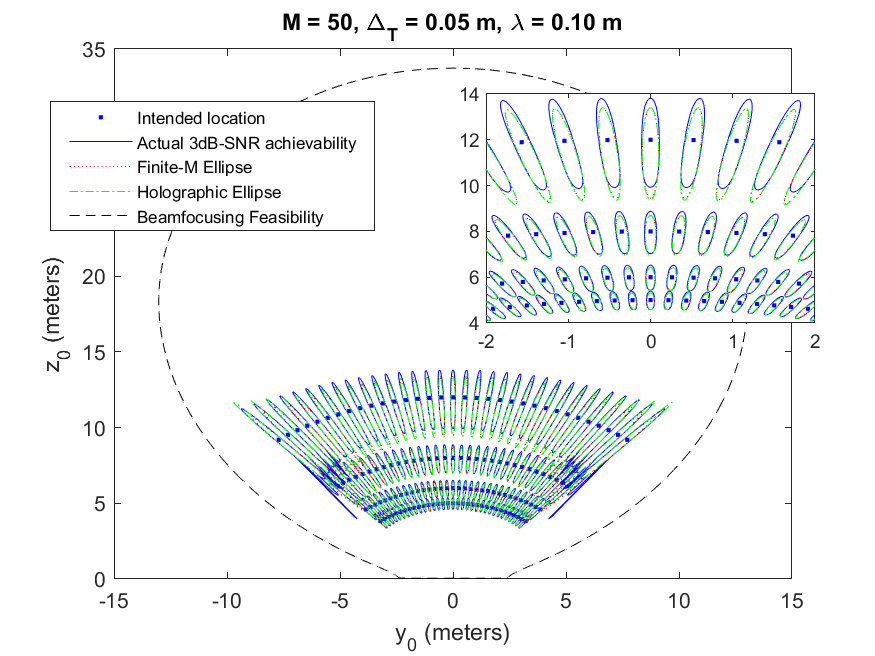}
    \caption{Representation of the region that achieves an SNR at least $3$dB below the optimum ($\kappa = 0.5$): true region (blue solid lines), finite-$M$ second order approximation (red dotted lines) and holographic approximation (green dash-dotted lines) for different intended locations (markers). }
    \label{fig:DeltaLambda2_near}
\end{figure}

Fig.~\ref{fig:DeltaLambda2_near} considers the situation where the intended receivers are relatively close to the transmit array. In this situation, the true $3$-dB SNR achievability regions are very well approximated by ellipses. Furthermore, both the finite-$M$ conic as its holographic approximation are quite close to the actual achievability region. The situation is somewhat different in Fig.~\ref{fig:DeltaLambda2_far}, which represents a situation where the intended receivers are further away from the ULA and closer to the holographic boundary of the region of beamfocusing feasibility (black dashed line). In this case, we observe that the actual $3$-dB SNR achievability regions (blue solid lines) are more complicated than just mere ellipses\footnote{Obviously, the shape of the achievability region becomes much closer to an ellipse if we consider higher values of $\kappa$ in (\ref{eq:quarticFiniteM}).}. However, they are still well-focused regions that are reasonably well approximated by ellipses near the intended receiver. Finally, Fig.~\ref{fig:DeltaLambda2_off} represents the situation where the intended receivers are located outside the holographic region for beamfocusing feasibility (black dashed lines). In this case, the actual $3$-dB SNR achievability regions are well approximated by hyperbolas and the ULA does no longer exhibit beamfocusing capabilities. In all the situations we can see that holographic approximations are quite accurate even for reasonably large values of the inter-element separation $\Delta_T$. 

\begin{figure}[t]
    \centering
    \includegraphics[width=\linewidth]{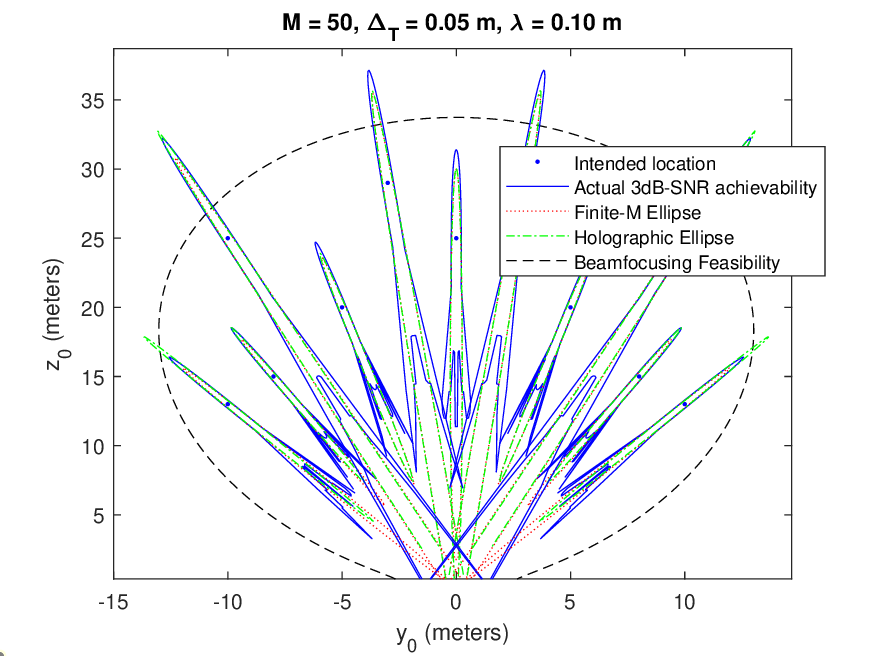}
    \caption{Representation of the region that achieves an SNR at least $3$dB below the optimum ($\kappa = 0.5$): true region (blue solid lines), finite-$M$ second order approximation (red dotted lines) and holographic approximation (green dash-dotted lines) for different intended locations (markers). }
    \label{fig:DeltaLambda2_far}
\end{figure}
\begin{figure}
    \centering
    \includegraphics[width=\linewidth]{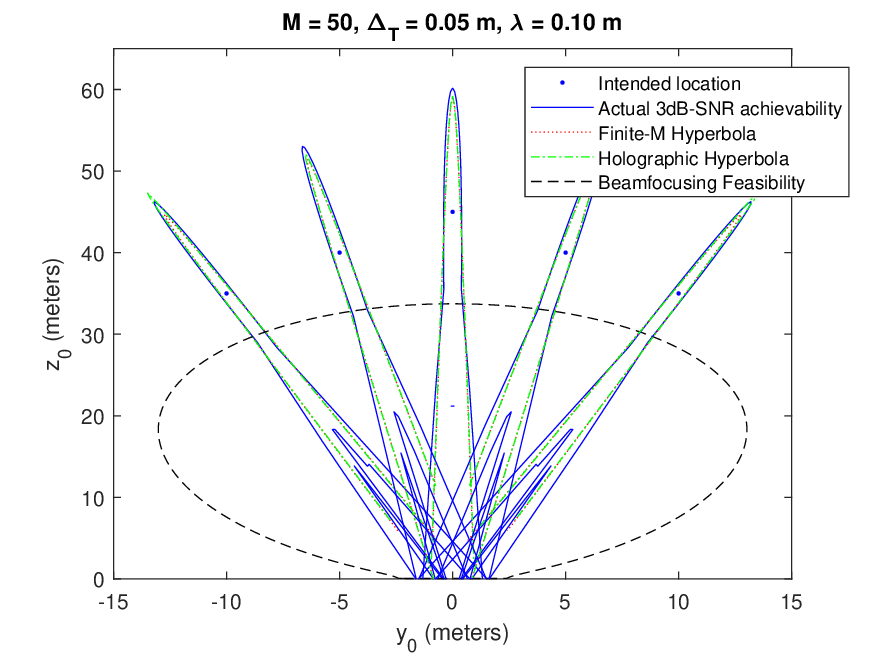}
    \caption{Representation of the region that achieves an SNR at least $3$dB below the optimum ($\kappa = 0.5$): true region (blue solid lines), finite-$M$ second order approximation (red dotted lines) and holographic approximation (green dash-dotted lines) for different intended locations (markers). }
    \label{fig:DeltaLambda2_off}
\end{figure}
\begin{figure}[ht]
    \centering
    \includegraphics[width=\linewidth]{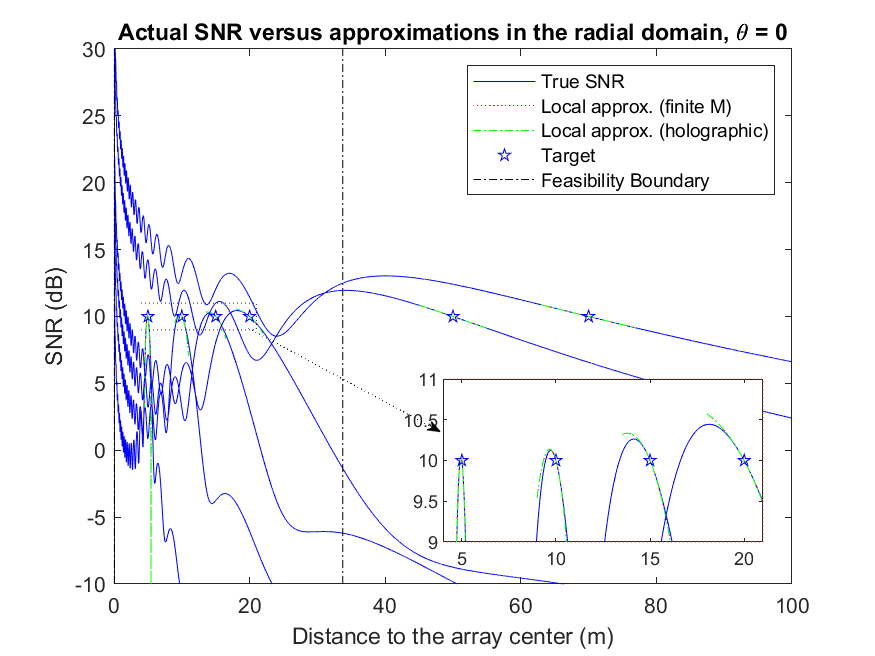}
    \caption{Actual achieved $\mathsf{SNR}$ for different targeted spatial points (markers) assuming $\mathsf{SNR}_0 =10$ dB and a receiver located at $\theta =0$ degrees. Local approximations are plotted in red dotted lines (finite $M$) and green dash-dotted lines (holographic). }
    \label{fig:beamfocusRadial0deg}
\end{figure}
\begin{figure}[ht]
    \centering
    \includegraphics[width=\linewidth]{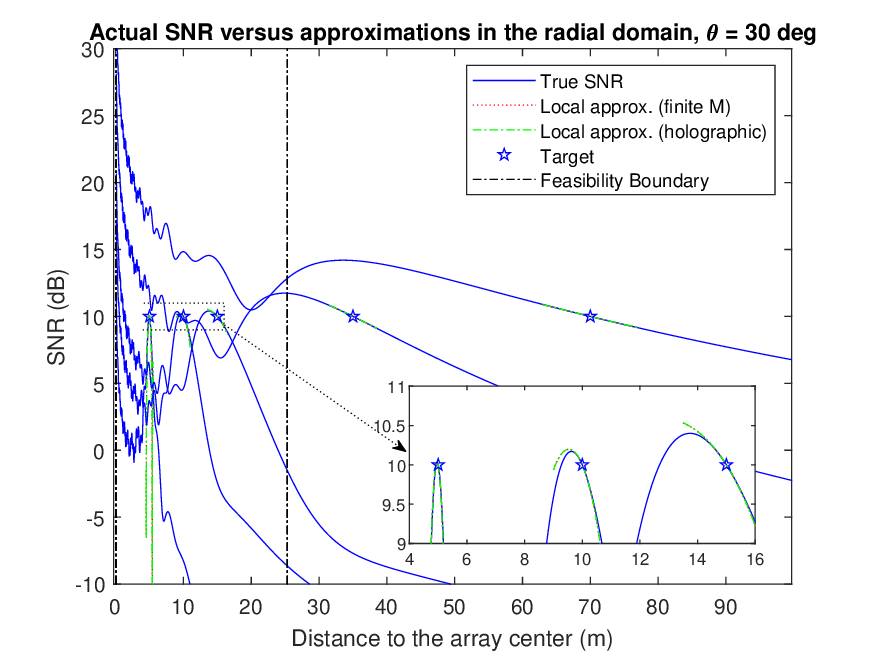}
    \caption{Actual achieved $\mathsf{SNR}$ for different targeted spatial points (markers) assuming $\mathsf{SNR}_0 =10$ dB and a receiver located at $\theta =30$ degrees. Local approximations are plotted in red dotted lines (finite $M$) and green dash-dotted lines (holographic). }
    \label{fig:beamfocusRadial30deg}
\end{figure}

Fig.~\ref{fig:beamfocusRadial0deg} and Fig.~\ref{fig:beamfocusRadial30deg} represent the actual achieved $\mathsf{SNR}$ when the transmitter is designed to guarantee a $\mathsf{SNR}_0 = 10$ dB at a certain radial location for an elevation equal to $\theta=0$ and $\theta =30$ degrees respectively. The intended locations are represented by markers and the local quadratic approximations for fixed $M$ and for the holographic regime are also shown in the same figure. Observe that the feasibility boundary can be used as the limit where beamfocusing can be safely used to multiplex signals in the radial domain. Beyond that point, the $\mathsf{SNR}$ loses its concavity and one cannot longer establish radial regions without mutual interference. 

\section{Conclusions}
The spatial region where beamfocusing is feasible by a ULA can be well characterized as the region where the SNR becomes locally concave. This means that the area close to the intended receiver that achieves a percentage of the maximum SNR can be asymptotically described as an ellipsoid. The holographic regime allows to characterize the beamfocusing feasibility region and provides some analytical insights into the shape and position of the coverage ellipsoids. According to the numerical analysis, the holographic description is quite accurate even in situations where the inter-element separation between consecutive elements of the ULA is moderately large. Furthermore, the asymptotic description is particularly useful in the region that is closer to the ULA, which corresponds to the area where beamfocusing is most effective. 
\appendices

\section{\label{sec:AppendixTaylor}Proof of Proposition \ref{prop:Taylor}}

We observe that we can express the complete $3\times3$ channel matrix
associated to the $m$th transmit antenna as $\mathbf{H}_{m}(\mathbf{r}
)=h_{m}(\mathbf{r})\mathbf{P}_{m}^{\perp}(\mathbf{r)}$ where $h_{m}
(\mathbf{r})$ is defined in (\ref{eq:defLowCaseh}). In this appendix, we will
use the short hand notation $\mathbf{r}_{m}=\mathbf{r}-\mathbf{p}_{m}$, where
we recall that $\mathbf{r}$ is the position vector of the receiver and
$\mathbf{p}_{m}$ is the position vector of the $m$th antenna. Using these
definitions, we can write $\mathbf{P}_{m}^{\perp}(\mathbf{r)}=\mathbf{I}
_{3}-\mathbf{r}_{m}\mathbf{r}_{m}^{H}/\left\Vert \mathbf{r}_{m}\right\Vert
^{2}$.

A Taylor expansion of the $i,j$th entry of the channel matrix (denoted as
$\left[  \mathbf{H}_{m}(\mathbf{r})\right]  _{i,j}$) can directly be expressed
as
\begin{multline*}
\left[  \mathbf{H}_{m}(\mathbf{r})\right]  _{i,j}   =\left[  \mathbf{H}
_{m}(\mathbf{r}_{0})\right]  _{i,j}+\left(  \mathbf{r}-\mathbf{r}_{0}\right)
^{T}\left[  \frac{d\left[  \mathbf{H}_{m}(\mathbf{r})\right]  _{i,j}
}{d\mathbf{r}}\right]  _{\mathbf{r=r}_{0}}\\
  +\frac{1}{2}\left(  \mathbf{r}-\mathbf{r}_{0}\right)  ^{T}\left[
\frac{d^{2}\left[  \mathbf{H}_{m}(\mathbf{r})\right]  _{i,j}}{d\mathbf{r}
d\mathbf{r}^{T}}\right]  _{\mathbf{r=r}_{0}}\left(  \mathbf{r}-\mathbf{r}
_{0}\right)  +\left[  \mathcal{R}_{m}(\mathbf{\bar{r}})\right]  _{i.j}
\end{multline*}
where $\mathcal{R}_{m}(\mathbf{\bar{r}})=O\left(  \left\Vert \mathbf{r}
-\mathbf{r}_{0}\right\Vert ^{3}\right)  $ is the reminder matrix and
$\mathbf{\bar{r}}$ is in the segment between $\mathbf{r}$ and $\mathbf{r}_{0}
$. Let us write $\mathbf{\Delta}_{r}=\left(  \mathbf{r}-\mathbf{r}_{0}\right)
$ and denote by $\mathcal{G}_{m}\left(  \mathbf{\Delta}_{r}\right)  $ and
$\mathcal{H}_{m}\left(  \mathbf{\Delta}_{r}\right)  $ the $3\times3$ matrices
defined as
\begin{align*}
\left[  \mathcal{G}_{m}\left(  \mathbf{\Delta}_{r}\right)  \right]  _{i,j}  &
=\left(  \mathbf{r}-\mathbf{r}_{0}\right)  ^{T}\left[  \frac{d\left[
\mathbf{H}_{m}(\mathbf{r})\right]  _{i,j}}{d\mathbf{r}}\right]  _{\mathbf{r=r}
_{0}}\\
\left[  \mathcal{H}_{m}\left(  \mathbf{\Delta}_{r}\right)  \right]  _{i,j}  &
=\frac{1}{2}\left(  \mathbf{r}-\mathbf{r}_{0}\right)  ^{T}\left[  \frac
{d^{2}\left[  \mathbf{H}_{m}(\mathbf{r})\right]  _{i,j}}{d\mathbf{r}
d\mathbf{r}^{T}}\right]  _{\mathbf{r=r}_{0}}\left(  \mathbf{r}-\mathbf{r}
_{0}\right)
\end{align*}
so that we can write $\mathbf{H}_{m}(\mathbf{r})=\mathbf{H}_{m}(\mathbf{r}
_{0})+\mathcal{G}_{m}\left(  \mathbf{\Delta}_{r}\right)  +\mathcal{H}
_{m}\left(  \mathbf{\Delta}_{r}\right)  +\mathcal{R}_{m}(\mathbf{\bar{r}})$
and where $\mathcal{G}_{m}\left(  \mathbf{\Delta}_{r}\right)  =O(\left\Vert
\mathbf{\Delta}_{r}\right\Vert )$ whereas $\mathcal{H}_{m}\left(
\mathbf{\Delta}_{r}\right)  =O(\left\Vert \mathbf{\Delta}_{r}\right\Vert
^{2})$ and $\mathcal{R}_{m}(\mathbf{\bar{r}})=O(\left\Vert \mathbf{\Delta}
_{r}\right\Vert ^{3})$. In (\ref{eq:TaylorChannel}), we define $\mathcal{G}
_{\mathrm{pol}}\left(  \mathbf{\Delta}_{r}\right)  $ and $\mathcal{H}
_{\mathrm{pol}}\left(  \mathbf{\Delta}_{r}\right)  $ from $\mathcal{G}
_{m}\left(  \mathbf{\Delta}_{r}\right)  $ and $\mathcal{H}_{m}\left(
\mathbf{\Delta}_{r}\right)  $ in the same way $\mathbf{H}_{\mathrm{pol}
}(\mathbf{r})$ is defined from $\mathbf{H}_{m}^{t_{\mathrm{pol}}\times
r_{\mathrm{pol}}}(\mathbf{r})$, and $\mathcal{P}_{\mathrm{pol}}=\mathcal{R}
_{\mathrm{pol}}/\left\Vert \mathbf{\Delta}_{r}\right\Vert ^{3}$ with
$\mathcal{R}_{\mathrm{pol}}$ built up from the reminder matrices
$\mathcal{R}_{m}(\mathbf{\bar{r}})$ in the same way. Next, we focus on the
characterization of these three matrices, namely $\mathcal{G}_{m}\left(
\mathbf{\Delta}_{r}\right)  $, $\mathcal{H}_{m}\left(  \mathbf{\Delta}
_{r}\right)  $ and $\mathcal{R}_{m}(\mathbf{\bar{r}})$. For the first two, we
provide closed form\ analytical expressions. For the reminder, we simply
reason how it can be uniformly bounded so that it is still of order
$O(\left\Vert \mathbf{\Delta}_{r}\right\Vert ^{3})$ in the holographic regime.
The following lemma provides some useful identities that will be used
throughout the derivations.

\begin{lemma}
\label{lem:Derivatives}Let $h_{m}(\mathbf{r})$ be defined as in
(\ref{eq:defLowCaseh}) and let $\mathbf{P}_{m}^{\perp}(\mathbf{r)}
=\mathbf{I}_{3}-\mathbf{P}_{m}(\mathbf{r)}$, $\mathbf{P}_{m}(\mathbf{r)=r}
_{m}\mathbf{r}_{m}^{H}/\left\Vert \mathbf{r}_{m}\right\Vert ^{2}$, where
$\mathbf{r}_{m}=\mathbf{r}-\mathbf{p}_{m}$. With some abuse of notation, let
$r_{k}$ denote the $k$th component of $\mathbf{r}$ and consider a
$3$-dimensional colum vector $\mathbf{a}$ with entries $a_{j}$. Then, we have
\begin{align}
\sum a_{k}\frac{\partial h_{m}(\mathbf{r})}{\partial r_{k}}  &  \mathbf{=}
\mathbf{-}\left(  1+\mathrm{j}\frac{2\pi}{\lambda}\left\Vert \mathbf{r}
_{m}\right\Vert \right)  \frac{\mathbf{a}^{T}\mathbf{r}_{m}}{\left\Vert
\mathbf{r}_{m}\right\Vert ^{2}}h_{m}(\mathbf{r})\label{eq:derh}\\
\sum a_{k}\frac{\partial}{\partial r_{k}}\mathbf{P}_{m}^{\perp}(\mathbf{r)}
&  \mathbf{=-}\frac{\mathbf{r}_{m}\mathbf{a}^{T}}{\left\Vert \mathbf{r}
_{m}\right\Vert ^{2}}\mathbf{P}_{m}^{\perp}(\mathbf{r)-P}_{m}^{\perp
}(\mathbf{r)}\frac{\mathbf{ar}_{m}^{T}}{\left\Vert \mathbf{r}_{m}\right\Vert
^{2}}\label{eq:derPort}\\
\sum a_{k}\frac{\partial\left\Vert \mathbf{r}_{m}\right\Vert }{\partial
r_{k}}  &  =\frac{\mathbf{a}^{T}\mathbf{r}_{m}}{\left\Vert \mathbf{r}
_{m}\right\Vert }\label{eq:derMod}\\
\sum a_{k}\frac{\partial}{\partial r_{k}}\frac{\mathbf{r}_{m}}{\left\Vert
\mathbf{r}_{m}\right\Vert ^{2}}  &  \mathbf{=}\left(  \mathbf{I}
_{3}-2\mathbf{P}_{m}(\mathbf{r)}\right)  \frac{\mathbf{a}}{\left\Vert
\mathbf{r}_{m}\right\Vert ^{2}}. \label{eq:derNormPos}
\end{align}

\end{lemma}

\subsection{Computation of the gradient $\mathcal{G}_{m}\left(  \mathbf{\Delta
}_{r}\right)  $}

Recall that $\mathbf{H}_{m}(\mathbf{r})=h_{m}(\mathbf{r})\mathbf{P}_{m}
^{\perp}(\mathbf{r)}$. A direct application of (\ref{eq:derh}
)-(\ref{eq:derPort}) in Lemma \ref{lem:Derivatives} shows that, using the
notation of the lemma,
\begin{align}
\sum_{k}a_{k}\frac{\partial}{\partial r_{k}}\mathbf{H}_{m}(\mathbf{r})  &
=-\left(  1+\mathrm{j}\frac{2\pi}{\lambda}\left\Vert \mathbf{r}_{m}\right\Vert
\right)  \frac{\mathbf{a}^{T}\mathbf{r}_{m}}{\left\Vert \mathbf{r}
_{m}\right\Vert ^{2}}\mathbf{H}_{m}(\mathbf{r})\label{eq:generalGradient}\\
&  -\frac{\mathbf{r}_{m}\mathbf{a}^{T}}{\left\Vert \mathbf{r}_{m}\right\Vert
^{2}}\mathbf{H}_{m}(\mathbf{r})-\mathbf{H}_{m}(\mathbf{r})\frac{\mathbf{ar}
_{m}^{T}}{\left\Vert \mathbf{r}_{m}\right\Vert ^{2}}.\nonumber
\end{align}
The closed form expression of the gradient matrix $\mathcal{G}_{m}\left(
\mathbf{\Delta}_{r}\right)  $ in (\ref{eq:defGm(DeltaT)}) can be obtained by particularizing the
above expression to $\mathbf{r=r}_{0}$ and $\mathbf{a}=\Delta_{r}
=\mathbf{r}-\mathbf{r}_{0}$, namely
\begin{align*}
\mathcal{G}_{m}\left(  \mathbf{\Delta}_{r}\right)   &  =-\left(
1+\mathrm{j}\frac{2\pi}{\lambda}\left\Vert \mathbf{r}_{m,0}\right\Vert
\right)  \frac{\Delta_{r}^{T}\mathbf{r}_{m,0}}{\left\Vert \mathbf{r}
_{m,0}\right\Vert ^{2}}\mathbf{H}_{m}(\mathbf{r}_{0})\\
&  -\frac{\mathbf{r}_{m,0}\Delta_{r}^{T}}{\left\Vert \mathbf{r}_{m,0}
\right\Vert ^{2}}\mathbf{H}_{m}(\mathbf{r}_{0})-\mathbf{H}_{m}(\mathbf{r}
_{0})\frac{\Delta_{r}\mathbf{r}_{m,0}^{T}}{\left\Vert \mathbf{r}
_{m,0}\right\Vert ^{2}}
\end{align*}
where now $\mathbf{r}_{m,0}=\mathbf{r}_{0}-\mathbf{p}_{m}$.

\subsection{Computation of the Hessian $\mathcal{H}_{m}\left(  \mathbf{\Delta
}_{r}\right)  $}

Starting from the expression in (\ref{eq:generalGradient}) and using the
identities in Lemma \ref{lem:Derivatives} we can directly obtain
\begin{multline}
\sum_{i,j}a_{i}b_{j}\frac{\partial^{2}\mathbf{H}_{m}(\mathbf{r})}{\partial
r_{i}\partial r_{j}}=\frac{h_{m}(\mathbf{r})}{\left\Vert \mathbf{r}
_{m}\right\Vert ^{2}} \Bigg[  \mathbf{\Xi}_{1}+\left(  1+\mathrm{j}\frac{2\pi
}{\lambda}\left\Vert \mathbf{r}_{m}\right\Vert \right)  \mathbf{\Xi}
_{2}+ \\+\left(  1+\mathrm{j}\frac{2\pi}{\lambda}\left\Vert \mathbf{r}
_{m}\right\Vert \right)  ^{2}\mathbf{\Xi}_{3}\Bigg]
\label{eq:secondOrderDer}
\end{multline}
where the matrices $\mathbf{\Xi}_{1}$, $\mathbf{\Xi}_{2}$ and $\mathbf{\Xi
}_{2}$ take the form
\begin{align*}
\mathbf{\Xi}_{1}  &  =-\mathbf{P}_{m}^{\perp}(\mathbf{r})\left(
\mathbf{ba}^{T}+\mathbf{ab}^{T}\right)  \mathbf{P}_{m}^{\perp}(\mathbf{r})\\
&  +\mathbf{P}_{m}(\mathbf{r})\left(  \mathbf{ba}^{T}+\mathbf{ab}^{T}\right)
\mathbf{P}_{m}^{\perp}(\mathbf{r})+\mathbf{P}_{m}^{\perp}(\mathbf{r})\left(
\mathbf{ab}^{T}+\mathbf{ba}^{T}\right)  \mathbf{P}_{m}(\mathbf{r})\\
&  +\left(  \mathbf{b}^{T}\mathbf{P}_{m}(\mathbf{r})\mathbf{a}\right)
\mathbf{P}_{m}^{\perp}(\mathbf{r})+2\left(  \mathbf{a}^{T}\mathbf{P}
_{m}^{\perp}(\mathbf{r})\mathbf{b}\right)  \mathbf{P}_{m}(\mathbf{r})\\
\mathbf{\Xi}_{2}  &  =\mathbf{P}_{m}(\mathbf{r})\left(  \mathbf{ab}
^{T}+\mathbf{ba}^{T}\right)  \mathbf{P}_{m}^{\perp}(\mathbf{r})+\mathbf{P}
_{m}^{\perp}(\mathbf{r})\left(  \mathbf{ba}^{T}+\mathbf{a\mathbf{b}}
^{T}\right)  \mathbf{P}_{m}(\mathbf{r})\\
&  -\mathbf{b}^{T}\mathbf{P}_{m}^{\perp}(\mathbf{r})\mathbf{aP}_{m}^{\perp
}(\mathbf{r})\\
\mathbf{\Xi}_{3}  &  =\mathbf{a}^{T}\mathbf{P}_{m}(\mathbf{r})\mathbf{bP}
_{m}^{\perp}(\mathbf{r}).
\end{align*}
Hence, the expression for $\mathcal{H}_{m}\left(  \mathbf{\Delta}_{r}\right)
$ can directly be obtained by particularizing the above expression to
$\mathbf{r=r}_{0}$ and $\mathbf{a=b=\Delta}_{r}$, so that
\begin{multline*}
\mathcal{H}_{m}\left(  \mathbf{\Delta}_{r}\right)  =\frac{h_{m}(\mathbf{r}
_{0})}{\left\Vert \mathbf{r}_{m,0}\right\Vert ^{2}} \Bigg[  \mathbf{\Xi}
_{1,0}+\left(  1+\mathrm{j}\frac{2\pi}{\lambda}\left\Vert \mathbf{r}
_{m,0}\right\Vert \right)  \mathbf{\Xi}_{2,0} \\ +\left(  1+\mathrm{j}\frac{2\pi
}{\lambda}\left\Vert \mathbf{r}_{m,0}\right\Vert \right)  ^{2}\mathbf{\Xi
}_{3,0}\Bigg]    
\end{multline*}
where we recall that $\mathbf{r}_{m,0}=\mathbf{r}_{0}-\mathbf{p}_{m}$ and
where the three matrices $\mathbf{\Xi}_{1,0}$, $\mathbf{\Xi}_{2,0}$,
$\mathbf{\Xi}_{3,0}$ take the form
\begin{align*}
\mathbf{\Xi}_{1,0}  &  =-2\mathbf{P}_{m}^{\perp}(\mathbf{r}_{0})\mathbf{\Delta
}_{r}\mathbf{\Delta}_{r}^{T}\mathbf{P}_{m}^{\perp}(\mathbf{r}_{0})\\
&  +2\mathbf{P}_{m}(\mathbf{r}_{0})\mathbf{\Delta}_{r}\mathbf{\Delta}_{r}
^{T}\mathbf{P}_{m}^{\perp}(\mathbf{r}_{0})+2\mathbf{P}_{m}^{\perp}
(\mathbf{r}_{0})\mathbf{\Delta}_{r}\mathbf{\Delta}_{r}^{T}\mathbf{P}
_{m}(\mathbf{r}_{0})\\
&  +\left(  \mathbf{\Delta}_{r}^{T}\mathbf{P}_{m}(\mathbf{r}_{0}
)\mathbf{\Delta}_{r}\right)  \mathbf{P}_{m}^{\perp}(\mathbf{r}_{0})+2\left(
\mathbf{\Delta}_{r}^{T}\mathbf{P}_{m}^{\perp}(\mathbf{r}_{0})\mathbf{\Delta
}_{r}\right)  \mathbf{P}_{m}(\mathbf{r}_{0})\\
\mathbf{\Xi}_{2,0}  &  =2\mathbf{P}_{m}(\mathbf{r}_{0})\mathbf{\Delta}
_{r}\mathbf{\Delta}_{r}^{T}\mathbf{P}_{m}^{\perp}(\mathbf{r}_{0}
)+2\mathbf{P}_{m}^{\perp}(\mathbf{r}_{0})\mathbf{\Delta}_{r}\mathbf{\Delta
}_{r}^{T}\mathbf{P}_{m}(\mathbf{r}_{0})\\
&  -\left(  \mathbf{\Delta}_{r}^{T}\mathbf{P}_{m}^{\perp}(\mathbf{r}
_{0})\mathbf{\Delta}_{r}\right)  \mathbf{P}_{m}^{\perp}(\mathbf{r}_{0})\\
\mathbf{\Xi}_{3,0}  &  =\left(  \mathbf{\Delta}_{r}^{T}\mathbf{P}
_{m}(\mathbf{r}_{0})\mathbf{\Delta}_{r}\right)  \mathbf{P}_{m}^{\perp
}(\mathbf{r}_{0}).
\end{align*}
Once again, we can use the above expression to obtain a closed form analytical
description of the product $\mathcal{H}_{\mathrm{pol}}\left(  \mathbf{\Delta
}_{r}\right)  \mathbf{H}_{\mathrm{pol}}^{H}\left(  \mathbf{r}_{0}\right)  $.
Indeed, using again the polarization selection matrix $\mathbf{E}
_{\mathrm{pol}}$ defined above, we can write

\subsection{Analysis of the reminder $\mathcal{R}_{m}(\mathbf{\bar{r}})$}

The objective of this section is to show that the spectral norm of the
reminder term $\mathcal{R}_{m}(\mathbf{\bar{r}})$ can be bounded by a quantity
of the type $\mathrm{P}\left(  \left\Vert \mathbf{\bar{r}-p}_{m}\right\Vert
^{-1}\right)  \left\Vert \mathbf{\Delta}_{r}\right\Vert ^{3}$, where
$\mathrm{P}\left(  \text{\textperiodcentered}\right)  $ is a polynomial with
positive coefficients independent of the geometry of the problem (a nice polynomial, see Remark~\ref{rem:niceConstantsPols}. In this
appendix, the actual value of $\mathrm{P}\left(  \text{\textperiodcentered
}\right)  $ may vary from one line to the next. First of all, a direct
application of (\ref{eq:derh})-(\ref{eq:derMod}) shows that
\[
\sum c_{k}\frac{\partial}{\partial r_{k}}\frac{h_{m}(\mathbf{r})}{\left\Vert
\mathbf{r}_{m}\right\Vert ^{2}}=\mathbf{-}\left(  3+\frac{2\pi\mathrm{j}
}{\lambda}\left\Vert \mathbf{r}_{m}\right\Vert \right)  \frac{\mathbf{c}
^{T}\mathbf{r}_{m}}{\left\Vert \mathbf{r}_{m}\right\Vert ^{4}}h_{m}
(\mathbf{r}).
\]
Consequently, applying the chain rule for derivatives in
(\ref{eq:secondOrderDer}) we find
\begin{align}
& \sum_{i,j,k}a_{i}b_{j}c_{k}\frac{\partial^{3}\mathbf{H}_{m}(\mathbf{r}
)}{\partial r_{i}\partial r_{j}\partial r_{k}}\label{eq:3rdDer} \\ &  =\mathbf{-}\left(
3+\frac{2\pi\mathrm{j}}{\lambda}\left\Vert \mathbf{r}_{m}\right\Vert \right)
\frac{\mathbf{c}^{T}\mathbf{r}_{m}}{\left\Vert \mathbf{r}_{m}\right\Vert ^{2}
}\sum_{i,j}a_{i}b_{j}\frac{\partial^{2}\mathbf{H}_{m}(\mathbf{r})}{\partial
r_{i}\partial r_{j}} \nonumber \\
&  +\frac{h_{m}(\mathbf{r})}{\left\Vert \mathbf{r}_{m}\right\Vert ^{2}}\left[
\mathbf{\tilde{\Xi}}_{1}+\left(  1+\frac{2\pi\mathrm{j}}{\lambda}\left\Vert
\mathbf{r}_{m}\right\Vert \right)  \mathbf{\tilde{\Xi}}_{2}+\left(
1+\frac{2\pi\mathrm{j}}{\lambda}\left\Vert \mathbf{r}_{m}\right\Vert \right)
^{2}\mathbf{\tilde{\Xi}}_{3}\right] \nonumber\\
&  +\mathrm{j}\frac{2\pi}{\lambda}\frac{h_{m}(\mathbf{r})\mathbf{c}
^{T}\mathbf{r}_{m}}{\left\Vert \mathbf{r}_{m}\right\Vert ^{3}}\left[
\mathbf{\Xi}_{2}+2\left(  1+\frac{2\pi\mathrm{j}}{\lambda}\left\Vert
\mathbf{r}_{m}\right\Vert \right)  \mathbf{\Xi}_{3}\right] \nonumber
\end{align}
where we have defined $\mathbf{\tilde{\Xi}}_{j}=\sum c_{k}\partial\mathbf{\Xi
}_{j}/\partial r_{k}$. Now, from the expression of $\mathbf{\Xi}_{1}$,
$\mathbf{\Xi}_{2}$, $\mathbf{\Xi}_{3}$ we can easily reason that $\left\Vert
\mathbf{\Xi}_{1}\right\Vert \leq9\left\Vert \mathbf{a}\right\Vert \left\Vert
\mathbf{b}\right\Vert $, $\left\Vert \mathbf{\Xi}_{2}\right\Vert
\leq5\left\Vert \mathbf{a}\right\Vert \left\Vert \mathbf{b}\right\Vert $ and
$\left\Vert \mathbf{\Xi}_{3}\right\Vert \leq\left\Vert \mathbf{a}\right\Vert
\left\Vert \mathbf{b}\right\Vert $. On the other hand, we have $\left\vert
h_{m}(\mathbf{r}_{0})\right\vert =\left\vert \xi/\lambda\right\vert
/\left\Vert \mathbf{r}_{m}\right\Vert $, which implies by the triangular
inequality that
\begin{multline*}
\left\Vert \sum_{i,j}a_{i}b_{j}\frac{\partial^{2}\mathbf{H}_{m}(\mathbf{r}
)}{\partial r_{i}\partial r_{j}}\right\Vert  \leq \\
\leq\frac{\left\vert \xi
/\lambda\right\vert }{\left\Vert \mathbf{r}_{m}\right\Vert ^{3}}\left[
15+\frac{14\pi}{\lambda}\left\Vert \mathbf{r}_{m}\right\Vert +\left(
\frac{2\pi}{\lambda}\right)  ^{2}\left\Vert \mathbf{r}_{m}\right\Vert
^{2}\right]  \left\Vert \mathbf{a}\right\Vert \left\Vert \mathbf{b}\right\Vert
\end{multline*}
whcih can be expressed as ${\left\Vert \mathbf{a}\right\Vert \left\Vert \mathbf{b}\right\Vert}/{\left\Vert \mathbf{r}_{m}\right\Vert } \mathrm{P} (  {\left\Vert \mathbf{r}_{m}\right\Vert ^{-1} }) $.
This shows that, since $\left\vert \mathbf{c}^{T}\mathbf{r}_{m}\right\vert
\leq\left\Vert \mathbf{c}\right\Vert \left\Vert \mathbf{r}_{m}\right\Vert $ by
Cauchy-Schwarz inequality, the modulus of the first term on the right hand
side of (\ref{eq:3rdDer}) is upper bounded by $\left\Vert \mathbf{a}
\right\Vert \left\Vert \mathbf{b}\right\Vert \left\Vert \mathbf{c}\right\Vert
\left\Vert \mathbf{r}_{m}\right\Vert ^{-1}\mathrm{P}(\left\Vert \mathbf{r}
_{m}\right\Vert ^{-1})$. A similar reasoning shows that the modulus of the
third term on the right hand side of (\ref{eq:3rdDer}) is upper bounded by $\left\Vert \mathbf{a}\right\Vert \left\Vert
\mathbf{b}\right\Vert \left\Vert \mathbf{c}\right\Vert$ times
\[
\frac{2\pi}{\lambda}\frac{\left\vert \xi/\lambda\right\vert }{\left\Vert
\mathbf{r}_{m}\right\Vert ^{3}}\left(  7+\frac{4\pi}{\lambda}\left\Vert
\mathbf{r}_{m}\right\Vert \right)   =\frac{1}{\left\Vert \mathbf{r}_{m}\right\Vert ^{2}}
\mathrm{P}\left( {\left\Vert \mathbf{r}_{m}\right\Vert ^{-1}}\right)  .
\]
It remains to bound the spectral norm of the matrices $\mathbf{\tilde{\Xi}
}_{j}$, $j=1,2,3$. We have from (\ref{eq:derPort}) and the Cauchy-Schwarz
inequality that $\left\Vert \sum c_{k}\partial\mathbf{P}_{m}(\mathbf{r)/}
\partial r_{k}\right\Vert \leq2\left\Vert \mathbf{c}\right\Vert \left\Vert
\mathbf{r}_{m}\right\Vert ^{-1}$. Therefore, by a direct application of the
chain rule we find that $\left\Vert \mathbf{\tilde{\Xi}}_{j}\right\Vert \leq
k_{j}\left\Vert \mathbf{a}\right\Vert \left\Vert \mathbf{b}\right\Vert
\left\Vert \mathbf{c}\right\Vert /\left\Vert \mathbf{r}_{m}\right\Vert $ where
$k_{1}=24$, $k_{2}=12$ and $k_{3}=4$. This shows that the second term on the
right hand side of (\ref{eq:3rdDer}) is upper bounded by $\left\Vert \mathbf{a}\right\Vert \left\Vert
\mathbf{b}\right\Vert \left\Vert \mathbf{c}\right\Vert$ times
\[
\frac{\left\vert \xi/\lambda\right\vert  }{\left\Vert
\mathbf{r}_{m}\right\Vert ^{4}}\left[  36+\frac{28\pi}{\lambda}\left\Vert
\mathbf{r}_{m}\right\Vert +\left(  \frac{2\pi}{\lambda}\right)  ^{2}\left\Vert
\mathbf{r}_{m}\right\Vert ^{2}\right]  =\frac{\mathrm{P}\left(  {\left\Vert
\mathbf{r}_{m}\right\Vert } ^{-1}\right)}{\left\Vert
\mathbf{r}_{m}\right\Vert ^{2}}  .
\]
With all this, we can conclude that
\[
\left\Vert \sum_{i,j,k}a_{i}b_{j}c_{k}\frac{\partial^{3}\mathbf{H}
_{m}(\mathbf{r})}{\partial r_{i}\partial r_{j}\partial r_{k}}\right\Vert
\leq\left\Vert \mathbf{a}\right\Vert \left\Vert \mathbf{b}\right\Vert
\left\Vert \mathbf{c}\right\Vert \mathrm{P}\left(  \frac{1}{\left\Vert
\mathbf{r}_{m}\right\Vert }\right)
\]
as we wanted to show.

\section{\label{sec:AppendixTaylorSNR}Proof of Corollary \ref{cor:TaylorSNR}}

We first observe that we can expand the product $\mathbf{H}_{m}\left(  \mathbf{r}\right)  \mathbf{H}_{m}^{H}\left(
\mathbf{r}\right) $ as
\begin{multline*}
\mathbf{H}_{m}\left(  \mathbf{r}\right)  \mathbf{H}_{m}^{H}\left(
\mathbf{r}\right)  =\mathbf{H}_{m}\left(  \mathbf{r}_{0}\right)
\mathbf{H}_{m}^{H}\left(  \mathbf{r}_{0}\right)  +\mathbf{H}_{m}\left(
\mathbf{r}_{0}\right)  \mathcal{G}_{m}^{H}\left(  \mathbf{\Delta}_{r}\right)
\\ +\mathcal{G}_{m}\left(  \mathbf{\Delta}_{r}\right)  \mathbf{H}_{m}^{H}\left(
\mathbf{r}_{0}\right) 
+\mathcal{G}_{m}\left(  \mathbf{\Delta}_{r}\right)  \mathcal{G}_{m}^{H}\left(
\mathbf{\Delta}_{r}\right) \\ +\mathcal{H}_{m}\left(  \mathbf{\Delta}
_{r}\right)  \mathbf{H}_{m}^{H}\left(  \mathbf{r}_{0}\right)  +\mathbf{H}
_{m}\left(  \mathbf{r}_{0}\right)  \mathcal{H}_{m}^{H}\left(  \mathbf{\Delta
}_{r}\right)  +\Xi_{m}
\end{multline*}
where we have introduced the error matrix 
\begin{align*}
\Xi_{m}  &  =\mathcal{G}_{m}\left(  \mathbf{\Delta}_{r}\right)  \mathcal{H}
_{m}^{H}\left(  \mathbf{\Delta}_{r}\right)  +\mathcal{H}_{m}\left(
\mathbf{\Delta}_{r}\right)  \mathcal{G}_{m}^{H}\left(  \mathbf{\Delta}
_{r}\right)   \\
&  +\mathcal{R}_{m}(\mathbf{\bar{r},\Delta}_{r})\left(  \mathbf{H}_{m}
^{H}\left(  \mathbf{r}_{0}\right)  +\mathcal{G}_{m}^{H}\left(  \mathbf{\Delta
}_{r}\right)  +\mathcal{H}_{m}^{H}\left(  \mathbf{\Delta}_{r}\right)  \right)
\\
&  +\left(  \mathbf{H}_{m}\left(  \mathbf{r}_{0}\right)  +\mathcal{G}
_{m}\left(  \mathbf{\Delta}_{r}\right)  +\mathcal{H}_{m}\left(  \mathbf{\Delta
}_{r}\right)  \right)  \mathcal{R}_{m}^{H}(\mathbf{\bar{r},\Delta}
_{r})\\
& +\mathcal{H}_{m}\left(  \mathbf{\Delta}_{r}\right)
\mathcal{H}_{m}^{H}\left(  \mathbf{\Delta}_{r}\right) +\mathcal{R}_{m}(\mathbf{\bar{r},\Delta}_{r})\mathcal{R}_{m}
^{H}(\mathbf{\bar{r},\Delta}_{r}).
\end{align*}
We only need to show that $\left\Vert \Xi_{m}\right\Vert <K\left\Vert
\mathbf{\Delta}_{r}\right\Vert ^{3}$ where $K$ is a nice constant. We recall
here that we can write $\left\Vert \mathbf{H}_{m}\left(  \mathbf{r}
_{0}\right)  \right\Vert \leq\mathrm{P}(\left\Vert \mathbf{r}_{m,0}\right\Vert
^{-1})$, $\left\Vert \mathcal{G}_{m}\left(  \mathbf{\Delta}_{r}\right)
\right\Vert \leq\mathrm{P}(\left\Vert \mathbf{r}_{m,0}\right\Vert
^{-1})\left\Vert \mathbf{\Delta}_{r}\right\Vert $,$\left\Vert \mathcal{H}
_{m}\left(  \mathbf{\Delta}_{r}\right)  \right\Vert \leq\mathrm{P}(\left\Vert
\mathbf{r}_{m,0}\right\Vert ^{-1})\left\Vert \mathbf{\Delta}_{r}\right\Vert
^{2}$ and $\mathcal{R}_{m}(\mathbf{\bar{r},\Delta}_{r})\leq\mathrm{P}
(\left\Vert \mathbf{\bar{r}}_{m,0}\right\Vert ^{-1})\left\Vert \mathbf{\Delta
}_{r}\right\Vert ^{2}$. Hence, using the fact that $\left\Vert \mathbf{A+B}
\right\Vert \leq\left\Vert \mathbf{A}\right\Vert +\left\Vert \mathbf{B}
\right\Vert $ and $\left\Vert \mathbf{AB}\right\Vert \leq\left\Vert
\mathbf{A}\right\Vert \left\Vert \mathbf{B}\right\Vert $ we can readily see
that $\left\Vert \Xi_{m}\right\Vert /\left\Vert \mathbf{\Delta}_{r}\right\Vert ^{3}$ is upper bounded by a term of the form
\[
\mathrm{P}(\left\Vert \mathbf{r}
_{m,0}\right\Vert ^{-1})+\mathrm{P}(\left\Vert \mathbf{r}_{m,0}\right\Vert
^{-1})\mathrm{P}(\left\Vert \mathbf{\bar{r}}_{m,0}\right\Vert ^{-1}
)+\mathrm{P}(\left\Vert \mathbf{\bar{r}}_{m,0}\right\Vert ^{-1})
\]
where we recall that the nice polynomials $\mathrm{P}($\textperiodcentered$)$
may take on a different value at each appearance. Now observe that $\left\Vert
\mathbf{r}_{m,0}\right\Vert >z_{0}$, $\left\Vert \mathbf{\bar{r}}
_{m,0}\right\Vert >\min\{z,z_0\}$ and so that $\left\Vert \Xi_{m}\right\Vert
\leq K\left\Vert \mathbf{\Delta}_{r}\right\Vert ^{3}$ where $K$ is a nice constant.

\section{Proof of Proposition \ref{prop:quadric}} \label{sec:AppendixQuadric}
It is well known that when $x_{0}=0$ the eigenvector associated to the maximum
eigenvalue of the matrix
\[
\frac{1}{2M+1}\mathbf{H}_{\mathrm{pol}}\left(  \mathbf{r}_{0}\right)
\mathbf{H}_{\mathrm{pol}}^{H}\left(  \mathbf{r}_{0}\right)
\]
is equal to the first column of a $3\times3$ identity matrix, denoted by $\mathbf{u}=\mathbf{e}_{1}$. This was proven in
\cite{mestre25ojsp} for $t_{\mathrm{pol}}=3$ and $t_{\mathrm{pol}}=2$ (the result for
$t_{\mathrm{pol}}=1$ is trivial since the above matrix has zeros everywhere
except for the $(1,1)$th entry). Using the fact that $\mathbf{u=e}_{1}$ we see
that only the first column of the two error matrices $\mathcal{G}
_{\mathrm{pol}}\left(  \Delta_{r}\right)  \mathbf{H}_{\mathrm{pol}}^{H}\left(
\mathbf{r}_{0}\right)  /\left(  2M+1\right)  $ and $\mathcal{H}_{\mathrm{pol}
}\left(  \Delta_{r}\right)  \mathbf{H}_{\mathrm{pol}}^{H}\left(
\mathbf{r}_{0}\right)  /\left(  2M+1\right)  $ come into play, which leads to
a considerable simplification of the above expression. To provide a closed
form expression for the matrix $\mathcal{G}_{\mathrm{pol}}\left(  \Delta
_{r}\right)  \mathbf{H}_{\mathrm{pol}}^{H}\left(  \mathbf{r}_{0}\right)  $,
let $\mathbf{E}_{\mathrm{pol}}$ denote a $3\times3$ diagonal selection matrix
that contains ones at the $t_{\mathrm{pol}}$ first diagonal entries and zeros
elsewhere. Using the expression of $\mathcal{G}_{m}\left(  \Delta_{r}\right)
$ in (\ref{eq:defGm(DeltaT)}) we can express
\begin{multline}
\mathcal{G}_{\mathrm{pol}}\left(  \Delta_{r}\right)  \mathbf{H}_{\mathrm{pol}
}^{H}\left(  \mathbf{r}_{0}\right)  =\sum_{m=-M}^{M}\mathcal{G}_{m}\left(
\mathbf{\Delta}_{r}\right)  \mathbf{E}_{\mathrm{pol}}\mathbf{H}_{m}^{H}\left(
\mathbf{r}_{0}\right) \\  =-\left\vert \frac{\xi}{\lambda}\right\vert ^{2}\sum_{m=-M}
^{M}\Bigg[ \frac{\mathbf{r}_{m,0}\Delta_{r}^{T}\mathbf{P}_{m}^{\perp}(\mathbf{r}
_{0})+\mathbf{P}_{m}^{\perp}(\mathbf{r}_{0})\Delta_{r}\mathbf{r}_{m,0}^{T}
}{\left\Vert \mathbf{r}_{m,0}\right\Vert ^{4}}\times \\ \times\mathbf{E}_{\mathrm{pol}
}\mathbf{P}_{m}^{\perp}(\mathbf{r}_{0})\label{eq:productGpolHpol}\\
%-\left\vert \frac{ \xi}{\lambda} \right\vert ^{2}\sum_{m=-M}^{M}
+ \frac{\mathbf{\Delta
}_{r}^{T}\mathbf{r}_{m,0}}{\left\Vert \mathbf{r}_{m,0}\right\Vert ^{4}}\left(
1+\mathrm{j}\frac{2\pi}{\lambda}\left\Vert \mathbf{r}_{m,0}\right\Vert
\right)  \mathbf{P}_{m}^{\perp}(\mathbf{r}_{0})\mathbf{E}_{\mathrm{pol}
}\mathbf{P}_{m}^{\perp}(\mathbf{r}_{0}) \Bigg]
\end{multline}
where we have used the fact that $\left\vert h_{m}(\mathbf{r}_{0})\right\vert
^{2}=\left\vert \xi/\lambda\right\vert ^{2}/\left\Vert \mathbf{r}
_{m,0}\right\Vert ^{2}$ and we recall that $\mathbf{r}_{m,0}=\mathbf{r}
_{0}-\mathbf{p}_{m}$. Regarding the expression for $\mathcal{H}_{\mathrm{pol}
}\left(  \Delta_{r}\right)  \mathbf{H}_{\mathrm{pol}}^{H}\left(
\mathbf{r}_{0}\right)  $, we can similarly use the expression of
$\mathcal{H}_{m}\left(  \Delta_{r}\right)  $ in (\ref{eq:defHm(DeltaT)}) to
see that
\begin{multline}
\mathcal{H}_{\mathrm{pol}}\left(  \mathbf{\Delta}_{r}\right)  \mathbf{H}
_{\mathrm{pol}}^{H}\left(  \mathbf{r}_{0}\right)    =\sum_{m=-M}
^{M}\mathcal{H}_{m}\left(  \mathbf{\Delta}_{r}\right)  \mathbf{E}
_{\mathrm{pol}}\mathbf{H}_{m}^{H}\left(  \mathbf{r}_{0}\right) \\
  =\sum_{m=-M}^{M}\frac{\left\vert \xi/\lambda\right\vert ^{2}}{\left\Vert
\mathbf{r}_{m,0}\right\Vert ^{4}}\Bigg[  \mathbf{\Xi}_{1,0}+\left(
1+\mathrm{j}\frac{2\pi}{\lambda}\left\Vert \mathbf{r}_{m,0}\right\Vert
\right)  \mathbf{\Xi}_{2,0} \\ +\left(  1+\mathrm{j}\frac{2\pi}{\lambda}\left\Vert
\mathbf{r}_{m,0}\right\Vert \right)  ^{2}\mathbf{\Xi}_{3,0}\Bigg]
\mathbf{E}_{\mathrm{pol}}\mathbf{P}_{m}^{\perp}(\mathbf{r}_{0}).\label{eq:productHesspolHpol}
\end{multline}
where the matrices $\mathbf{\Xi}_{j,0}$ are defined in the statement of
Proposition \ref{prop:Taylor}. To further simplify these expressions, 
% we
% recall the form of the orthogonal projection matrix $\mathbf{P}_{m}
% ^{\mathbf{\perp}}(\mathbf{r}_{0})$ when $x_{0}=0$, that is
% \[
% \mathbf{P}_{m}^{\mathbf{\perp}}(\mathbf{r}_{0})=\frac{1}{\left\Vert
% \mathbf{r}_{0}-\mathbf{p}_{m}\right\Vert ^{2}}\left[
% \begin{array}
% [c]{ccc}
% \left\Vert \mathbf{r}_{0}-\mathbf{p}_{m}\right\Vert ^{2} & 0 & 0\\
% 0 & z_{0}^{2} & -\left(  y_{0}-m\Delta_{T}\right)  z_{0}\\
% 0 & -\left(  y_{0}-m\Delta_{T}\right)  z_{0} & \left(  y_{0}-m\Delta
% _{T}\right)  ^{2}
% \end{array}
% \right]
% \]
we use the fact that $\mathbf{P}_{m}^{\mathbf{\perp}}(\mathbf{r}
_{0})\mathbf{e}_{1}=\mathbf{e}_{1}$ and $\mathbf{P}_{m}(\mathbf{r}
_{0})\mathbf{e}_{1}=0$. Furthermore, $\mathbf{E}_{\mathrm{pol}}\mathbf{e}
_{1}=\mathbf{e}_{1}$\ because the polarization along the $x$-axis is assumed
to be always employed. These properties allow to simplify the expression in
(\ref{eq:productGpolHpol}) into
\begin{equation}
\frac{1}{2M+1}\mathcal{G}_{\mathrm{pol}}\left(  \Delta_{r}\right)
\mathbf{H}_{\mathrm{pol}}^{H}\left(  \mathbf{r}_{0}\right)  \mathbf{u}
=\left\vert \xi/\lambda\right\vert ^{2}\mathbb{G}\mathbf{\Delta}_{r}
\label{eq:simplifiedGradient}
\end{equation}
where
\[
\mathbb{G}=\left[
\begin{array}
[c]{ccc}
0 & \bar{s}_{M}^{(3)}+\mathrm{j}\frac{2\pi}{\lambda}\bar{s}_{M}^{(2)} &
-s_{M}^{(4)}z_{0}-\mathrm{j}\frac{2\pi}{\lambda}s_{M}^{(3)}z_{0}\\
\bar{s}_{M}^{(3)} & 0 & 0\\
-z_{0}s_{M}^{(4)} & 0 & 0
\end{array}
\right]
\]
Likewise, using the expression in (\ref{eq:productHesspolHpol}) we immediately
obtain
\begin{align*}
& \mathcal{H}_{\mathrm{pol}}\left(  \mathbf{\Delta}_{r}\right)  \mathbf{H}
_{\mathrm{pol}}^{H}\left(  \mathbf{r}_{0}\right)  \mathbf{u} \\ &  =-\Delta
_{x}\left\vert \frac{\xi}{\lambda}\right\vert ^{2}\sum_{m=-M}^{M}\frac{1}{\left\Vert
\mathbf{r_{0}-p}_{m}\right\Vert ^{4}} \mathbf{P}_{m}^{\perp}  \mathbf{\Delta}_{r}\\
&+ \Delta
_{x}\left\vert \frac{\xi}{\lambda}\right\vert ^{2}\sum_{m=-M}^{M}\frac{1}{\left\Vert
\mathbf{r_{0}-p}_{m}\right\Vert ^{4}} \left(  2+\mathrm{j}\frac{2\pi}{\lambda}\left\Vert
\mathbf{r_{0}-p}_{m}\right\Vert \right)  \mathbf{P}_{m}(\mathbf{r}
_{0})  \mathbf{\Delta}_{r}\\
&  +\frac{\left\vert \xi/\lambda\right\vert ^{2}}{2}\sum_{m=-M}^{M}\left[
1+\left(  1+\mathrm{j}\frac{2\pi}{\lambda}\left\Vert \mathbf{r_{0}-p}
_{m}\right\Vert \right)  ^{2}\right] \times \\ 
& \times \frac{\mathbf{\Delta}_{r}^{T}
\mathbf{P}_{m}(\mathbf{r}_{0})\mathbf{\Delta}_{r}}{\left\Vert \mathbf{r_{0}
-p}_{m}\right\Vert ^{4}}\mathbf{u}\\
&  -\frac{\left\vert \xi/\lambda\right\vert ^{2}}{2}\sum_{m=-M}^{M}\left(
1+\mathrm{j}\frac{2\pi}{\lambda}\left\Vert \mathbf{r_{0}-p}_{m}\right\Vert
\right)  \frac{\mathbf{\Delta}_{r}^{T}\mathbf{P}_{m}^{\perp}(\mathbf{r}
_{0})\mathbf{\Delta}_{r}}{\left\Vert \mathbf{r_{0}-p}_{m}\right\Vert ^{4}
}\mathbf{u.}
\end{align*}
At this point, we can use the expression of $\mathbf{P}_{m}^{\mathbf{\perp}
}(\mathbf{r}_{0})$ so that
\begin{multline}
\frac{1}{2M+1}\mathcal{H}_{\mathrm{pol}}\left(  \mathbf{\Delta}_{r}\right)
\mathbf{H}_{\mathrm{pol}}^{H}\left(  \mathbf{r}_{0}\right)  \mathbf{u}  = \\
=\Delta_{x}\left\vert \xi/\lambda\right\vert ^{2}\mathbb{H}_{1}\mathbf{\Delta
}_{r}+\frac{\left\vert \xi/\lambda\right\vert ^{2}}{2}\mathbf{\Delta}_{r}
^{T}\mathbb{H}_{2}\mathbf{\Delta}_{r}\mathbf{u} \label{eq:simplifiedHessian}
\end{multline}
where we have introduced the two matrices
\begin{align*}
\mathbb{H}_{1}  &  =\left[
\begin{array}
[c]{ccc}
-s_{M}^{(4)} & 0 & 0\\
0 & 2s_{M}^{(4)}-3z_{0}^{2}s_{M}^{(6)} & -3\bar{s}_{M}^{(5)}z_{0}\\
0 & -3\bar{s}_{M}^{(5)}z_{0} & -s_{M}^{(4)}+3z_{0}^{2}s_{M}^{(6)}
\end{array}
\right] \\ & +\frac{2\pi\mathrm{j}}{\lambda}\left[
\begin{array}
[c]{ccc}
0 & 0 & 0\\
0 & s_{M}^{(3)}-z_{0}^{2}s_{M}^{(5)} & -\bar{s}_{M}^{(4)}z_{0}\\
0 & -\bar{s}_{M}^{(4)}z_{0} & z_{0}^{2}s_{M}^{(5)}
\end{array}
\right] \\
\mathbb{H}_{2}  &  =\left[
\begin{array}
[c]{ccc}
-s_{M}^{(4)} & 0 & 0\\
0 & 2s_{M}^{(4)}-3z_{0}^{2}s_{M}^{(6)} & -3\bar{s}_{M}
^{(5)}z_{0}\\
0 & -3z_{0}\bar{s}_{M}^{(5)} & 3s_{M}^{(6)}z_{0}^{2}-s_{M}^{(4)}
\end{array}
\right]
\\ &+ \left(\frac{2\pi}{\lambda}\right)^2 \left[
\begin{array}
[c]{ccc}
0 & 0 & 0\\
0 &  z_{0}^{2}s_{M}^{(4)} - s_{M}^{(2)}  & \bar{s}_{M}^{(3)}z_{0}\\
0 & \bar{s}_{M}^{(3)}z_{0} & -z_{0}^{2}s_{M}^{(4)}
\end{array} \right] \\
&  +\frac{2\pi\mathrm{j}}{\lambda}\left[
\begin{array}
[c]{ccc}
-s_{M}^{(3)} & 0 & 0\\
0 & 2s_{M}^{(3)}-3z_{0}^{2}s_{M}^{(5)} & -3\bar{s}_{M}^{(4)}z_{0}\\
0 & -3\bar{s}_{M}^{(4)}z_{0} & 3z_{0}^{2}s_{M}^{(5)}-s_{M}^{(3)}
\end{array}
\right]  .
\end{align*}
Now, using (\ref{eq:simplifiedGradient}) and (\ref{eq:simplifiedHessian}) into
(\ref{eq:TaylorSNRsimple}), we can conclude that
\begin{align*}
\mathsf{SNR}(\mathbf{r})  &  =\mathsf{SNR}(\mathbf{r}_{0}) +\mathsf{SNR}
(\mathbf{r}_{0})\frac{1}{s_{M}^{(2)}}\left(  \mathbf{\Delta}_{r}^{T}
\mathbb{G}^{H}\mathbf{u}+\mathbf{u}^{H}\mathbb{G}\mathbf{\Delta}_{r}\right) \\
&  +\mathsf{SNR}(\mathbf{r}_{0})\frac{1}{s_{M}^{(2)}}  \Delta_{x}\left(
\mathbf{\Delta}_{r}^{T}\mathbb{H}_{1}^{H}\mathbf{u}+\mathbf{u}^{H}
\mathbb{H}_{1}\mathbf{\Delta}_{r}\right) \\
&  +\mathsf{SNR}(\mathbf{r}_{0})\mathbf{\Delta}_{r}^{T} \left[ 
\frac{\mathbb{G}^{H}\mathbb{G}}{( s_{M}^{(2)})  ^{2}} + 
\frac{\mathbb{H}_{2}+\mathbb{H}_{2}^{H}}{2 s_{M}^{(2)}} 
\right]\mathbf{\Delta}_{r}
+\epsilon_{M}
\end{align*}
where we have used the fact that $\lambda_{\max}=\left\vert \xi/\lambda
\right\vert ^{2}s_{M}^{(2)}$ and $\mathsf{SNR}(\mathbf{r}_{0})=\bar{P}
\lambda_{\max}/\sigma^{2}$. Using the expression of $\mathbb{G}$,
$\mathbb{H}_{1}$ and $\mathbb{H}_{2}$ we can compactly express the above
expression as in (\ref{eq:SNRtaylorx0=0}) where
\begin{align*}
\mathfrak{m}_{M}  &  =-\frac{1}{s_{M}^{(2)}}\operatorname{Re}\left[
\mathbb{G}\right]  \mathbf{u}\\
\mathcal{M}_{M}  &  =-\frac{1}{\left(  s_{M}^{(2)}\right)  ^{2}}
\mathrm{Re}\mathbb{G}^{H}\mathbb{G}-\frac{1}{2s_{M}^{(2)}}\left(  \mathbb{H}_{2}+\mathbb{H}_{2}^{H}\right)
\\&- \frac{1}{s_{M}^{(2)}}\left(
\mathbb{H}_{1}^{H}\mathbf{uu}^{H}+\mathbf{uu}^{H}\mathbb{H}_{1}\right)
\end{align*}
which correspond to the expressions given in the statement of the proposition.
The only remaining point is to prove that $\gamma_{M}^{(1)}\geq0$. To see
this, we can use the Cauchy-Schwarz inequality to show that
\begin{align}
s_{M}^{(2k)}s_{M}^{(2l)}  &  \geq\left(  s_{M}^{(k+l)}\right)  ^{2}
\label{eq:ineqss1}\\
s_{M}^{2k}\left(  s_{M}^{2(l-1)}-z_{0}^{2}s_{M}^{(2l)}\right)   &  \geq\left(
\bar{s}_{M}^{k+l-1}\right)  ^{2} \label{eq:ineqss2}
\end{align}
for any $k,l\in\mathbb{N}$. To see that $\gamma_{M}^{(1)}>0$, use
(\ref{eq:ineqss1})-(\ref{eq:ineqss2}) with $k=1$, $l=3$ so that $s_{M}
^{(2)}s_{M}^{(6)}\geq(s_{M}^{(4)})^{2}$ and $s_{M}^{(2)}(s_{M}^{(4)}-z_{0}
^{2}s_{M}^{(6)})\geq(\bar{s}_{M}^{(3)})^{2}$. Applying these two inequalities
consecutively, we find $\gamma_{M}^{(1)}=(3s_{M}^{(2)}s_{M}^{(4)}-(\bar{s}
_{M}^{(3)})^{2}-z_{0}^{2}(s_{M}^{(4)})^{2})/(s_{M}^{(2)})^{2}\geq(3s_{M}
^{(2)}s_{M}^{(4)}-(\bar{s}_{M}^{(3)})^{2}-z_{0}^{2}s_{M}^{(2)}s_{M}
^{(6)})/(s_{M}^{(2)})^{2}\geq2(s_{M}^{(2)}s_{M}^{(4)})/(s_{M}^{(2)})^{2}>0$.

% % use section* for acknowledgment
% \section*{Acknowledgment}


\ifCLASSOPTIONcaptionsoff
  \newpage
\fi



\bibliographystyle{IEEEtran}
%\balance
\bibliography{IEEEabrv,bibtex/bib/holographic}  

\end{document}

%% file: figures/ULA.tex
%\tdplotsetmaincoords{120}{-40};

\begin{tikzpicture}[>={latex},scale=0.4,every node/.style={scale=0.4}]
%\begin{tikzpicture}[>={latex},tdplot_main_coords]
%\begin{tikzpicture}[>={latex}]

\def\tripole at (#1,#2){
\node (A) [scale = 0.3, cylinder, shape border rotate=90, shape aspect=.5, draw = red, minimum height= 60, minimum width=1] at (#1+0,#2+0.35) {};
\node (A) [scale = 0.3, cylinder, shape border rotate=90, shape aspect=.5, draw = red, minimum height= 60, minimum width=1] at (#1+0,#2-0.35) {};
\node (A) [scale = 0.3,cylinder, shape border rotate=0, shape aspect=.5, draw = blue, minimum height= 60, minimum width=1] at (#1+0.35,#2+0) {};
\node (A) [scale = 0.3,cylinder, shape border rotate=0, shape aspect=.5, draw = blue, minimum height= 60, minimum width=1] at (#1-0.35,#2+0) {};
\node (A) [scale = 0.3,cylinder, rotate=43, shape aspect=.5, draw = black!60!green, minimum height= 60, minimum width=1] at (#1+0.35/1.25,#2+0.32/1.25) {};
\node (A) [scale = 0.3,cylinder, rotate=43, shape aspect=.5, draw = black!60!green, minimum height= 60, minimum width=1] at (#1-0.35/1.2,#2-0.33/1.2) {};
}

\def\tripoleGros at (#1,#2){
\node (A) [cylinder, shape border rotate=90, shape aspect=.5, draw = red, minimum height= 60, minimum width=1] at (#1+0,#2+1.3) {};
\node (A) [cylinder, shape border rotate=90, shape aspect=.5, draw = red, minimum height= 60, minimum width=1] at (#1+0,#2-1.3) {};
\node (A) [cylinder, shape border rotate=0, shape aspect=.5, draw = blue, minimum height= 60, minimum width=1] at (#1+1.3,#2+0) {};
\node (A) [cylinder, shape border rotate=0, shape aspect=.5, draw = blue, minimum height= 60, minimum width=1] at (#1-1.3,#2+0) {};
\node (A) [cylinder, rotate=43, shape aspect=.5, draw = black!60!green, minimum height= 60, minimum width=1] at (#1+1.3/1.25,#2+1.1/1.25) {};
\node (A) [cylinder, rotate=43, shape aspect=.5, draw = black!60!green, minimum height= 60, minimum width=1] at (#1-1.3/1.2,#2-1.5/1.2) {};
}

\draw [->] (0,0) node (v4) {} -- (12,0) node[right] {\LARGE $z$};
\draw [->] (0,-1) -- (0,5) node[left] {\LARGE $x$};
\draw [<-] (-1.3*4,-1.2*4)  node[below] {\LARGE $y$} -- (1.3*4,1.2*4) ;
\draw [<->] (-1.3*3.1-2.3,-1.2*3.1)  -- (1.3*3.1-2.3,1.2*3.1) ;
\node [right] at (-2.25,1.25) {\LARGE $2L$};

 \foreach \x in {-3,...,3}
       	\tripole at (1.3*\x,1.2*\x);

%\draw (10.5,0) node (v3) {} -- (10.5,4.5) coordinate (v1) {} -- (v1) -- (8.5,3) node (v2) {} -- (v2) -- (v3) -- (v2) -- (v4);
%\draw (10.5,0) node (v3) {} -- (12.5,2.5) coordinate (v1) {}; % -- (v1) -- (8.5,3) node (v2) {} ;
%\draw (8.5,3) node (v2) {} -- (v2) -- (v3) -- (v2) -- (v4);

% 
%\draw[<->] (4.5,1.5) .. controls (5,1) and (5.5,0.5) .. (5.5,0);
%\draw[<->] (11.5,1.5) .. controls (11,2) and (10,2) .. (9.5,1.5);
%\node at (5.5,1) {\LARGE $\theta$};
%\node at (10,2.5)  {\LARGE $\varphi$};

%\draw [<->] (-4.5,-4) -- (-4.5,-2.3);
%\node at (-5,-3) {\LARGE $\Delta_T$};

\draw [<->] (-2.2,-1.2) -- (-0.9,0);
\node [right] at (-2.5,-0.3) {\LARGE $\Delta_T$};

\node [right] at (-5.5,-3) {\LARGE $M$};
\node [right] at (-5.6,-1.7) {\LARGE $(M-1)$};
\node [right] at (0.0,3.5) {\LARGE $-(M-1)$};
\node [right] at (2.5,4.5) {\LARGE $-M$};
 
%\draw [<->] (6.5,0.4) node (v5) {} -- (-1.3,-6.8) node (v6) {};
%\node [right] at (3,-3.5) {\LARGE $2M\Delta_T$};
% \draw [<->] (-6.5,-0.3) -- (-6.5,-6.8);
% \node [rotate=90] at (-7,-3.5) {\LARGE $2K\Delta_T$};
 
%\draw [dotted,very thick] (3.5,0.4) -- (v5);
%\draw [dotted,very thick] (-6.5,-6.8) --   (v6);
%\draw [dotted,very thick] (-6.5,-0.35) -- (-3.9,-0.35);
\node [right] at (5.8,1.6) {\LARGE $D$};

\tripoleGros at (7, -3.5); 

\begin{scope} [scale=1.5, shift={(-5.5,1.5)}]
\draw[red] (10.15,-3.62) .. controls (10.35,-4.1) and (10.9,-4.1) .. (11.3,-4.2);
\draw[red] (10.15,-4) .. controls (10.15,-3.7) and (10.5,-4.3) .. (11.3,-4.35);
\draw [blue] (9.99,-3.83) .. controls (10.34,-3.74) and (9.65,-3.54) .. (9.2,-3.25);
\draw [blue] (10.37,-3.84) .. controls (10.14,-3.81) and (9.65,-3.3) .. (9.2,-3.1);
\draw [black!60!green] (10.36,-3.68) .. controls (9.95,-4.05) and (9.7,-4.05) .. (9.35,-4.1);
\draw [black!60!green] (9.95,-4.15) .. controls (10.05,-4.05) and (9.8,-4.05) .. (9.35,-4.25);
\draw [blue] (9.2,-3) rectangle node {RF} (8.5,-3.4);
\draw [red] (11.85,-4.1) rectangle node {RF} (11.3,-4.45);
\draw [black!60!green] (8.65,-4) rectangle node {RF} (9.35,-4.35);

\end{scope}

\draw[->] (v4) -- (9.6,2) node (v1) {};
 \tripole at (9.6,2);

\draw (3.5,0) .. controls (4,0.5) and (4,0.5) .. (4,0.8);
\node at (4.2,0.5) {\LARGE $\theta$};

%\draw[dashed] (v1) -- (3.9,3.6);
%\draw[dashed] (-3.9,-3.6) -- (v1) -- cycle;
%\draw (7.9,1.3) .. controls (7.7,1.8) and (7.8,2.1) .. (8.2,2.4);

%\node at (7.5,2) {\LARGE $\Gamma$};
\end{tikzpicture}